\documentclass{llncs}
\pdfoutput=1

\usepackage[ruled,linesnumbered,noend]{algorithm2e} 
\usepackage{amsmath} 
\usepackage{amssymb} 
\usepackage{cases}

\usepackage{pgfplots}
\usepackage{xspace}

\usepackage{graphicx} 
\usepackage{microtype}

\usepackage{wrapfig} 

\newcommand{\dist}{\operatorname{dist}}

\newcommand{\ie}{i.\,e.\xspace}

\newcommand{\eg}{e.\,g.\xspace}

\newcommand{\etal}{et al.\xspace}

\newcommand{\wrt}{w.\,r.\,t.\xspace}

\newcommand{\Pro}[1]{\mathbf{Pr} \left[\,#1\,\right]}

\newcommand{\pointset}{\ensuremath{\mathrm{P}}\xspace}
\newcommand{\neighborset}{\ensuremath{\mathrm{N}}}
\newcommand{\candidateset}{\ensuremath{\mathrm{Candidates}}}
\newcommand{\cellset}{\ensuremath{\mathrm{Cells}}}
\newcommand{\overhangset}{\ensuremath{\mathrm{Overhang}}\xspace}

\newcommand{\ringset}{\ensuremath{\varsigma}\xspace}

\newcommand{\old}[1]{}

\DeclareMathOperator*{\argmin}{arg\,min}
\begin{document}

\title{Querying Probabilistic Neighborhoods \\ in Spatial Data Sets Efficiently}
\institute{\{moritz.looz-corswarem, meyerhenke\}@kit.edu\\Institute of Theoretical Informatics, Karlsruhe Institute of Technology (KIT), Germany}
\author{Moritz von Looz \and Henning Meyerhenke}


\maketitle

\begin{abstract}
The probability that two spatial objects establish some kind of mutual connection often depends on their proximity.
To formalize this concept, we define the notion of a \emph{probabilistic neighborhood}:
Let $P$ be a set of $n$ points in $\mathbb{R}^d$,  $q \in \mathbb{R}^d$ a query point, $\dist$ a distance metric, and $f : \mathbb{R}^+ \rightarrow [0,1]$ a monotonically decreasing function.
Then, the probabilistic neighborhood $N(q, f)$ of $q$ with respect to $f$ is 
a random subset of $P$ and each point $p \in P$ belongs to $N(q,f)$ with probability $f(\dist(p,q))$.
Possible applications include query sampling and the simulation of probabilistic spreading phenomena, as well as other scenarios where the probability of a connection between two entities decreases with their distance.
We present a fast, sublinear-time query algorithm to sample probabilistic neighborhoods from planar point sets.
For certain distributions of planar $P$, we prove that our algorithm answers a query in $O((|N(q,f)| + \sqrt{n})\log n)$ time with high probability.
In experiments this yields a speedup over pairwise distance probing of at least one order of magnitude, even for rather small data sets with $n=10^5$ and also for other point distributions not covered by the theoretical results.

\end{abstract}

\section{Introduction}
\label{sec:introduction}
%

In many scenarios, connections between spatial objects are not certain but probabilistic, with the probability depending on the distance between them:
The probability that a customer shops at a certain physical store shrinks with increasing distance to it.
In disease simulations, if the social interaction graph is unknown but locations are available, disease transmission 
can be modeled as a random process with infection risk decreasing with distance.
Moreover, the wireless connections between units in an ad-hoc network are fragile and collapse more frequently with higher 
distance. 

For these and similar scenarios, we define the notion of a \emph{probabilistic neighborhood} in spatial data sets:
Let a set $P$ of $n$ points in $\mathbb{R}^d$, a query point $q \in \mathbb{R}^d$, a distance metric $\dist$, 
and a monotonically decreasing function $f : \mathbb{R}^+ \rightarrow [0,1]$ be given.
Then, the probabilistic neighborhood $N(q, f)$ of $q$ with respect to $f$ is 
a random subset of $P$ and each point $p \in P$ belongs to $N(q,f)$ with probability $f(\dist(p,q))$.
A straightforward query algorithm for sampling a probabilistic neighborhood
would iterate over each point $p \in P$ and sample for each whether it is included in $\neighborset(q, f)$.
This has a running time of $\Theta(n \cdot d)$ per query point, which is prohibitive for repeated queries in large data sets.
Thus we are interested in a faster algorithm for such a \emph{probabilistic neighborhood query} (PNQ, spoken as ``pink'').
We restrict ourselves to the planar case in this work, but the algorithmic principle is generalizable to higher dimensions.

While the linear-time approach has appeared before in the literature for a particular application~\cite{Aldecoa2015} 
(without formulating the problem as a PNQ explicitly), we are not aware of previous work performing more efficient PNQs 
with an index structure. For example, the probabilistic quadtree introduced by Kraetzschmar \etal~\cite{kraetzschmar2004probabilistic} 
is designed to store probabilistic occupancy data and gives deterministic results.
Other range queries related to (yet different from) our work as well as deterministic index structures are described in Section~\ref{sec:related}.

\paragraph{Contributions.}
We develop, analyze, implement, and evaluate an index structure and a query algorithm that 
together provide fast probabilistic neighborhood queries in the Euclidean and hyperbolic plane.
Our key data structure for these fast PNQs is a polar quadtree which we adapt from our previous work~\cite{Looz2015HRG}.
Preprocessing for quadtree construction requires $O(n \log n)$ time with high probability\footnote{We say ``with high probability'' (whp) when 
referring to a probability $\geq 1- 1/n$ for sufficiently large $n$.} (whp).

To answer PNQs, we first present a simple query algorithm (Section~\ref{sec:baseline}).
We then improve its time complexity by treating whole subtrees as so-called virtual leaves, see Section~\ref{sec:subtree-aggr}.
As shown by our detailed theoretical analysis, the improved algorithm yields a query
time complexity of $O((|\neighborset(q, f)| + \sqrt{n})\log n)$ whp to find a probabilistic neighborhood $\neighborset(q,f)$ among $n$ points, for $n$ sufficiently large. This is sublinear if the returned neighborhood $\neighborset(q, f)$ is of size $o(n / \log n)$
-- an assumption we consider reasonable for most applications.
For our theoretical results to hold, the quadtree structure needs to be able to partition the distribution of the point positions in $\pointset$, \ie not all of the probability mass may be concentrated on a single point or line.
In our case of polar quadtrees, this is achieved if the distribution is continuous, integrable, rotationally invariant 
with respect to the origin and non-zero only for a finite area.

Experimental results are shown in Section~\ref{sec:applications}: We apply our query algorithm to generate random graphs
in the hyperbolic plane~\cite{Krioukov2010} in subquadratic time. Graphs with millions of edges can now be generated 
within a few minutes sequentially.
This yields an acceleration of at least one order of magnitude in practice compared to a reference implementation~\cite{Aldecoa2015} that uses linear-time queries. Compared to our previous work on graph generation~\cite{Looz2015HRG}, 
our new algorithm is able to generate a more extensive model.
Even if the distribution of a given point set $P$ is unknown in practice, running times are fast:
As an example of probabilistic spreading behavior, we simulate a simple disease spreading mechanism on real population density geodata.
In this scenario, our fast PNQs are at least two orders of magnitude faster than linear-time queries.

\section{Preliminaries}
\label{sec:prelim}

\subsection{Notation}
\label{sub:notation}
\newcommand{\hyperbolic}{\ensuremath{\mathbb{H}}}
\newcommand{\Euclidean}{\ensuremath{\mathbb{E}}}
Let the input be given as set $\pointset$ of $n$ points. 
The points in $\pointset$ are distributed in a disk $\mathbb{D}_R$ of radius $R$ in the hyperbolic or Euclidean plane, the distribution is given by a probability density function $j(\phi, r)$ for an angle $\phi$ and a radius $r$.
Recall that, for our theoretical results to hold, we require $j$ to be known, continuous and integrable.
Furthermore, $j$ needs to be rotationally invariant -- meaning that $j(\phi_1, r) = j(\phi_2, r)$ for any radius $r$ and any two angles $\phi_1$ and $\phi_2$ -- and positive within  $\mathbb{D}_R$,
so that $j(r) > 0 \Leftrightarrow r < R$.
Due to the rotational invariance, $j(\phi, r)$ is the same for every $\phi$ and we can write $j(r)$.
Likewise, we define $J(r)$ as the indefinite integral of $j(r)$ and normalize it so that $J(R) = 1$ (also implying $J(0) = 0$). The value $J(r)$ then gives the fraction of probability mass inside radius $r$.

For the distance between two points $p_1$ and $p_2$, we use $\dist_\hyperbolic{}(p_1, p_2)$ for the hyperbolic and $\dist_\Euclidean{}(p_1, p_2)$ for the Euclidean case.
We may omit the index if a distinction is unnecessary.
As mentioned, a point $p$ is in the probabilistic neighborhood of query point $q$ with probability $f(\dist(p, q))$.
Thus, a \emph{query pair} consists of a query point $q$ and a function $f : \mathbb{R}^+ \rightarrow [0,1]$ that maps distances to probabilities.
The function $f$ needs to be monotonically decreasing but may be discontinuous. Note that $f$ can be 
defined differently for each query.
The query result, the probabilistic neighborhood of $q$ \wrt $f$, is denoted by the set $\neighborset(q,f) \subseteq P$.

For the algorithm analysis, we use two additional sets for each query $(q,f)$:
\begin{itemize}
 \item $\candidateset(q,f)$: neighbor candidates examined when executing such a query,
 \item $\cellset(q, f)$: quadtree cells examined during execution of the query.
\end{itemize}
Note that the sets $\neighborset(q,f), \candidateset(q,f)$ and $\cellset(q,f)$ are probabilistic, thus theoretical results about their size are usually only with high probability.

\subsection{Related Work}
\label{sec:related}
\paragraph{Fast deterministic range queries.}
Numerous index structures for fast range queries on spatial data exist.
Many such index structures are based on trees or variations 
thereof, see Samet's book~\cite{Samet:2005:FMM:1076819} for a comprehensive overview.
I/O efficient worst case analysis is usually performed using the EM model,
see \eg~\cite{Arge:2012:ISD:2367574.2367575}. In more applied settings, average-case performance is of
higher importance, which popularized R-trees or newer variants thereof, \eg~\cite{Kamel:1994:HRI:645920.673001}.
Concerning (balanced) quadtrees for spatial dimension $d$, it is known that queries require $O(d \cdot n^{1-1/d})$ time
(thus $O(\sqrt{n})$ in the planar case)~\cite[Ch.~1.4]{Samet:2005:FMM:1076819}.
Regarding PNQs our algorithm matches this query complexity up to a logarithmic factor.
Yet note that, since for general $f$ and $\dist$ in our scenario all points in the set $P$ could be neighbors, 
data structures for deterministic queries cannot solve a PNQ efficiently without adaptations.

Hu et al.~\cite{Hu2014independent} give a query sampling algorithm for one-dimensional data that,
given a set $P$ of n points in $\mathbb{R}$,  an interval $q = [x,y]$ and an integer, $t \geq 1$, returns $t$ elements uniformly sampled from $P \cap q$.
They describe a structure of $O(n)$ space that answers a query in $O(\log n + t)$ time and supports updates in $O(\log n)$ time.
While also offering query sampling, PNQs differ from the problem considered by Hu et al. in two aspects: We consider two dimensions instead of one and our sampling probabilities are not necessarily uniform,
but can be set by the user by a distance-dependent function.

\paragraph{Range queries on uncertain data.}
During the previous decade probabilistic queries \emph{different} from PNQs have become popular.
The main scenarios can be put into two categories~\cite{pei2008query}: (i) Probabilistic databases contain entries
that come with a specified confidence (\eg sensor data whose accuracy is uncertain) and
(ii) objects with an uncertain location, \ie the location is specified by a probability distribution.
Both scenarios differ under typical and reasonable assumptions from ours: 
Queries for uncertain data are usually formulated to return \emph{all} points in the neighborhood
whose confidence/probability exceeds a certain threshold~\cite{kriegel2007probabilistic},
or computing points that are possibly nearest neighbors~\cite{agarwal2013nearest}.

In our model, in turn, the choice of inclusion of a point $p$ is a random choice for every different $p$. In particular, 
depending on the probability distribution, \emph{all}
nodes in the plane can have positive probability to be part of some other's neighborhood.
In the related scenarios this would only be true with extremely small confidence values or extremely
large query circles.

\paragraph{Applications in fast graph generation.}
One application for PNQs as introduced in Section~\ref{sec:introduction} is the 
hyperbolic random graph model by Krioukov \etal~\cite{Krioukov2010}. The $n$ graph nodes are represented by points 
thrown into the hyperbolic plane at random\footnote{The probability density in the polar model depends only on radii $r$ and $R$ as well as a growth parameter $\alpha$ and is given by $g(r) := \alpha\frac{\sinh(\alpha r)}{\cosh(\alpha R)-1} $.} and
two nodes are connected by an edge with a probability that decreases with
the distance between them. An implementation of this generative model is
available~\cite{Aldecoa2015}, it performs $\Theta(n^2)$ neighborhood tests. 
Bringmann et al. provide an algorithm to generate hyperbolic random graphs in expected linear time~\cite{bringmann2015geometric}; to our knowledge no implementation of it exists yet.

In previous work we designed a  
generator~\cite{Looz2015HRG} faster than~\cite{Aldecoa2015} for a restricted model; it runs in $O((n^{3/2}+m) \log n)$
time whp for the whole graph with $m$ edges. The range queries discussed there are facilitated by a quadtree 
which supports only deterministic queries. Consequently, the queries result in unit-disk graphs in the hyperbolic plane 
and can be considered as a special case of the current work (a step function $f$ with values 0 and 1 results in a deterministic query).

Our major technical inspiration for enhancing the quadtree for probabilistic neighborhoods is the work of 
Batagelj and Brandes~\cite{batagelj2005efficient}. They were the first to present a random sampling 
method to generate Erd\H{o}s-R\'{e}nyi-graphs with $n$ nodes and $m$ edges in $O(n + m)$ time complexity.
Faced with a similar problem of selecting each of $n$ elements with a constant probability $p$, they designed an 
efficient algorithm (see Algorithm~\ref{algo:fast-batagelji-brandes} in Appendix~\ref{sec:fast-batagelji-brandes}).
Instead of sampling each element separately, they use random jumps of length $\delta(p)$,
$\delta(p) = \ln(1-\mathit{rand}) / \ln(1-p)$, with $\mathit{rand}$ being a random number uniformly distributed in $[0,1)$.

\subsection{Quadtree Specifics}
\label{sub:prelim-quadtree}
Our key data structure is a polar region quadtree in the Euclidean or hyperbolic plane.
While they are less suited to higher dimensions as for example k-d-trees, the complexity is comparable in the plane.
For the (circular) range queries we discuss, quadtrees have the significant advantage of a bounded aspect ratio:
A cell in a k-d-tree might extend arbitrarily far in one direction, rendering theoretical guarantees about the area affected by the query circle difficult to impossible.
In contrast, the region covered by a quadtree cell is determined by its position and level.

We mostly reuse our previous definition~\cite{Looz2015HRG} of the quadtree:
A node in the quadtree is defined as a tuple $(\mathrm{min}_\phi, \mathrm{max}_\phi, \mathrm{min}_r, \mathrm{max}_r)$
with \(\mathrm{min}_\phi \leq \mathrm{max}_\phi\) and \(\mathrm{min}_r \leq \mathrm{max}_r\).
It is responsible for a point $p = (\phi_p, r_p)$ exactly if
$(\mathrm{min}_\phi \leq \phi_p < \mathrm{max}_\phi)$ and $(\mathrm{min}_r \leq r_p < \mathrm{max}_r)$.
We call the region represented by a particular quadtree node its quadtree \emph{cell}.
The quadtree is parametrized by its radius $R$, the $\mathrm{max}_r$ of the root cell.
If the probability distribution $j$ is known (which we assume for our theoretical results), 
we set the radius $R$ to $\argmin_r J(r) = 1$, \ie to the minimum radius that contains the full probability mass.
If only the points are known, the radius is set to include all of them.
While in this latter case the complexity analysis of Section~\ref{sec:baseline} and~\ref{sec:subtree-aggr} does not hold,
fast running times in practice can still be achieved (see Section~\ref{sec:applications}).

\section{Baseline Query Algorithm}
\label{sec:baseline}
We begin the main technical part by describing adaptations in the quadtree construction as well as
a baseline query algorithm. This latter algorithm introduces the main idea, but is asymptotically not faster than the 
straightforward approach. In Section~\ref{sec:subtree-aggr} it is then refined to support faster queries.

\subsection{Quadtree Construction}
\label{sub:qt-augment}
\begin{wrapfigure}[18]{R}{0.4\linewidth}
 \centering
\vspace{-3\baselineskip}
 \includegraphics[width=0.9\linewidth]{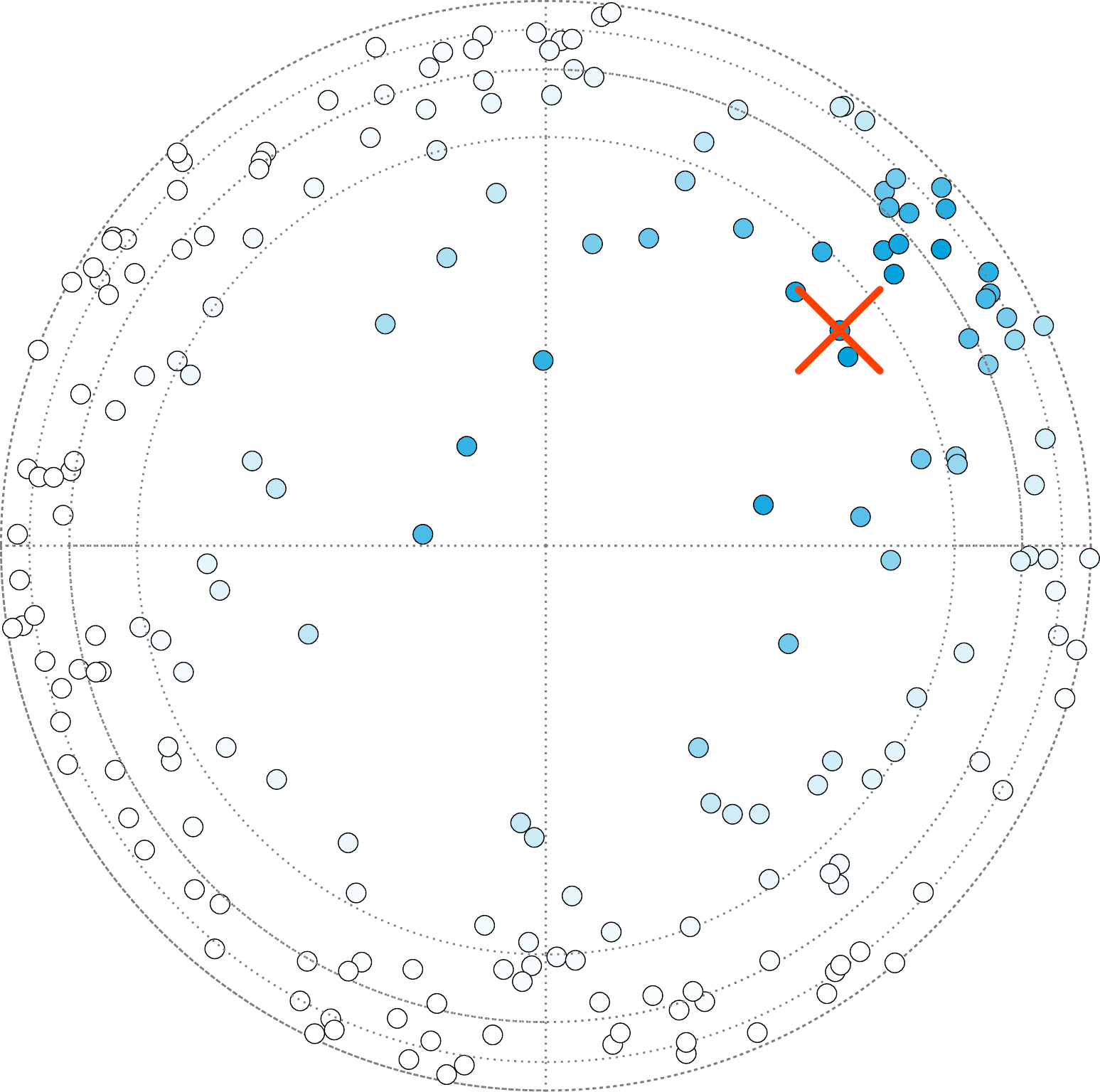}
 \caption{Query over 200 points in a polar hyperbolic quadtree, with $f(d) := 1/(e^{(d-7.78)}+1)$ and the query point $q$ marked by a red cross.
Points are colored according to the probability that they are included in the result. Blue represents a high probability, white a probability of zero.}
\label{fig:visualization-point-probabilities}
\end{wrapfigure}

At each quadtree node $v$, we store the size of the subtree rooted there.
We then generalize the rule for node splitting to handle point distributions $j$ as defined
in Section~\ref{sub:notation}:
As is usual for quadtrees, a leaf cell $c$ is split into four children when it exceeds its fixed capacity.
Since our quadtree is polar, this split happens once in the angular and once in the radial direction.
Due to the rotational symmetry of $j$, splitting in the angular direction is straightforward as the angle range is halved: $\mathrm{mid}_\phi := \frac{\mathrm{max}_\phi+\mathrm{min}_\phi}{2}$.
For the radial direction, we choose the splitting radius to result in an equal division of probability mass.
The total probability mass in a ring delimited by $\min_r$ and $\max_r$ is $J(\mathrm{max}_r) - J(\mathrm{min}_r)$.
Since $j(r)$ is positive for $r$ between $R$ and 0, the restricted function $J|_{[0,R]}$ defined above is a bijection.
The inverse $(J|_{[0,R]})^{-1}$ thus exists and we set the splitting radius $\mathrm{mid}_r$ to $(J|_{[0,R]})^{-1}\left(\frac{J(\mathrm{max}_r) + J(\mathrm{min}_r)}{2}\right)$.

Figure~\ref{fig:visualization-point-probabilities} visualizes a point distribution on a hyperbolic disk with 200 points
and Figure~\ref{fig:quadtree-example} its corresponding quadtree.

\begin{figure}[b]
 \includegraphics[width=\linewidth]{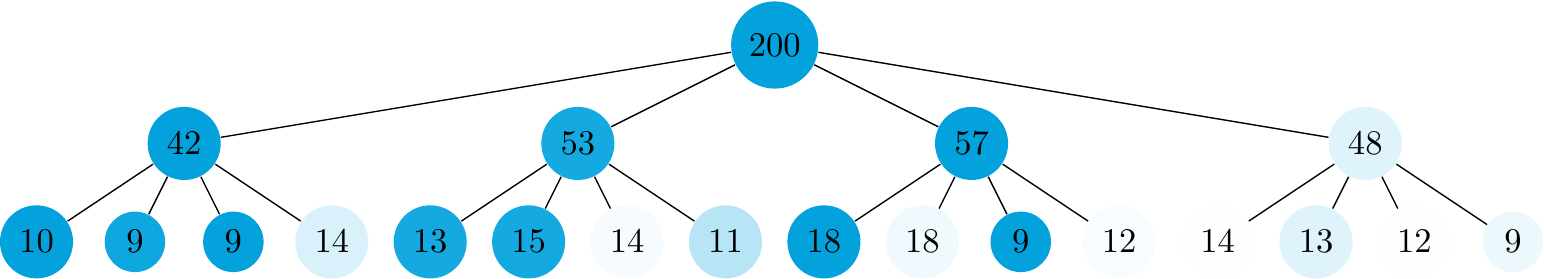}
 \caption{Visualization of the data structure used in Figure~\ref{fig:visualization-point-probabilities}.
 Quadtree nodes are colored according to the upper probability bound for points contained in them.
 The color of a quadtree node $c$ is the darkest possible shade (dark = high probability) of any point contained in the subtree rooted at $c$.
 Each node is marked with the number of points in its subtree.}
  \label{fig:quadtree-example}
 \end{figure}

Two results on quadtree properties help to establish the time complexity of quadtree operations.
They are generalized versions of our previous work~\cite[Lemmas~1 and~2]{Looz2015HRG} and state that each quadtree cell contains the same expected number of points and that the quadtree height is $O(\log n)$ whp (proofs in Appendix~\ref{sec:basic-qt-proofs}).

\begin{lemma}
Let $\mathbb{D}_R$ be a hyperbolic or Euclidean disk of radius $R$, $j$ a probability distribution on $\mathbb{D}_R$ which fulfills the properties defined in Section~\ref{sub:notation}, $p$ a point in $\mathbb{D}_R$ which is sampled from $j$, and $T$ be a polar quadtree on $\mathbb{D}_R$.
Let $C$ be a quadtree cell at depth $i$. Then, the probability that $p$ is in $C$ is $4^{-i}$.
\label{lemma:node-cell-probabilities}
\end{lemma}
\begin{lemma}
Just as in Lemma~\ref{lemma:node-cell-probabilities}, 
let $\mathbb{D}_R $ be a hyperbolic or Euclidean disk of radius $R$, $j$ a probability distribution on $\mathbb{D}_R$ which fulfills the properties defined in Section~\ref{sub:notation}, and $T$ be a polar quadtree on $\mathbb{D}_R$.
The expected number of  nodes in $T$ is then in $O(n)$.
\label{lemma:bound-number-quadtree-cells}
\end{lemma}
\begin{proposition}
 \label{thm:quadtree-height}
Let $\mathbb{D}_R$ and $j$ be as in Lemma~\ref{lemma:node-cell-probabilities}.
Let $T$ be a polar quadtree on $\mathbb{D}_R$ constructed to fit $j$.
Then, for $n$ sufficiently large, $\mathrm{height}(T) \in O(\log n)$ whp.
\end{proposition}

A direct consequence from the results above and our previous work~\cite{Looz2015HRG} is the preprocessing
time for the quadtree construction. The generalized splitting rule and storing the subtree sizes only change constant factors.

\begin{corollary}
\label{cor:qt-construction}
Since a point insertion takes $O(\log n)$ time whp, constructing a quadtree on $n$ points distributed as 
in Section~\ref{sub:notation} takes $O(n \log n)$ time whp.
\end{corollary}

\subsection{Algorithm}
\label{sub:baseline-algo}
The baseline version of our query (Algorithm~\ref{algo:quadnode-probabilistic}) has unfortunately a time complexity of $\Theta(n)$, but serves as a foundation for the fast version (Section~\ref{sec:subtree-aggr}).
It takes as input a query point $q$, a function $f$ and a quadtree cell $c$.
Initially, it is called with the root node of the quadtree and recursively descends the tree.
The algorithm returns a point set $\neighborset(q,f) \subseteq P$ with
\begin{equation}
\label{eq:probneigh}
\Pro{p \in \neighborset(q,f)} =  f(\mathrm{dist}(q,p)).
\end{equation}

\begin{algorithm}[tb]
 \KwIn{query point $q$, prob.\ function $f$, quadtree node $c$}
 \KwOut{probabilistic neighborhood of \textit{q}}
 $\neighborset = \{\}$\;
  $\underline{b} = \dist(q, c$)\;\label{line:distanceLB} \tcc{Distance between point and cell}
 $\overline{b}$=$f(\underline{b}$)\;\label{line:probUB} \tcc{Since $f$ is monotonically decreasing, a lower bound for the distance gives an upper bound $\overline{b}$ for the probability.}
 $s$ = number of points in $c$\;
 \If{$c$ is not leaf}{\label{line:notleafcell}
 \tcc{internal node: descend, add recursive result to local set}
 \For{child $\in$ children($c$)}{
   add getProbabilisticNeighborhood($q$, $f$, child) to \neighborset \;\label{line:calling-child}
 }
}
 \Else{\tcc{leaf case: apply idea of Batagelj and Brandes~\cite{batagelj2005efficient}} 
 \For{i=0; $i < s$ ; i++}{\label{line:loop-iteration} \label{line:candidate-loop}
 $\delta = \ln(1-\mathit{rand}) / \ln(1-\overline{b})$\;\label{line:deltaskip}
 $i$ += $\delta$\;\label{line:jumptarget}
 \If{$i \geq s$}{
 break\;
 }
 $\mathit{prob}$ = $f(\dist(q$, c.points[$i$]))/$\overline{b}$\;\label{line:candidate-confirmation}
 add c.points[$i$] to $\neighborset$ with probability $\mathit{prob}$\label{line:add-confirmed}
 }
 }
 \Return{$\neighborset$}
 \caption{QuadNode.getProbabilisticNeighborhood}
 \label{algo:quadnode-probabilistic}
\end{algorithm}

Algorithm~\ref{algo:quadnode-probabilistic} descends the quadtree recursively until it reaches the leaves.
Once a leaf $l$ is reached, a lower bound $\underline{b}$ for the distance between the query point $q$ and all the points in $l$ is computed (Line~\ref{line:distanceLB}). Such distance calculations are detailed in Appendix~\ref{sec:distance}.
Since $f$ is monotonically decreasing, this lower bound for the distance gives an upper bound $\overline{b}$ for the probability that a given point in $l$ is a member of the returned point set (Line~\ref{line:probUB}).
This bound is used to select \emph{neighbor candidates} in a similar manner as Bategelj and Brandes~\cite{batagelj2005efficient}: 
In Line~\ref{line:deltaskip}, a random number of vertices is skipped, so that every vertex in $l$ is selected as a neighbor candidate with probability $\mathrm{\overline{b}}$.
The actual distance $\mathrm{dist}(q,a)$ between a candidate $a$ and the query point $q$ is at least $\mathrm{\underline{b}}$ and the probability of $a \in \neighborset(q, f)$ thus at most $\mathrm{\overline{b}}$.
For each candidate, this actual distance $\mathrm{dist}(q,a)$ is then calculated and a neighbor candidate is confirmed as a neighbor with probability $f(\mathrm{dist}(q,a))/\mathrm{\overline{b}}$ in Line~\ref{line:candidate-confirmation}.

Regarding correctness and time complexity of Algorithm~\ref{algo:quadnode-probabilistic}, we can state:

\begin{proposition}
Let $T$ be a quadtree as defined above, $q$ be a query point and $f : \mathbb{R}^+ \rightarrow [0,1]$ a monotonically decreasing function
which maps distances to probabilities.
The probability that a point $p$ is returned by a PNQ ($q, f$) from Algorithm~\ref{algo:quadnode-probabilistic} is $f(\text{dist}(q,p))$,
independently from whether other points are returned.
\label{lemma:independent-correct-probabilities}
\end{proposition}

\begin{proposition}
\label{cor:basic-time}
Let $T$ be a quadtree with $n$ points.
The running time of Algorithm~\ref{algo:quadnode-probabilistic} per query on $T$ is $\Theta(n)$ in expectation.
\end{proposition}
The proofs can be found in Appendices~\ref{sub:proof-lemma-indep} and \ref{sec:proof-cor-basic-time}.

\section{Queries in Sublinear Time by Subtree Aggregation}
\label{sec:subtree-aggr}
%
%
One reason for the linear time complexity of the baseline query is the fact that every quadtree node is visited.
To reach a sublinear time complexity, we thus aggregate subtrees into \emph{virtual leaf cells} whenever doing so reduces the number of examined cells and does not increase the number of candidates too much.

To this end, let $S$ be a subtree starting at depth $l$ of a quadtree $T$.
During the execution of Algorithm~\ref{algo:quadnode-probabilistic}, a lower bound $\underline{b}$ for the distance between $S$ and the query point $q$ is calculated,
yielding also an upper bound $\overline{b}$ for the neighbor probability of each point in $S$.
At this step, it is possible to treat $S$ as a \emph{virtual leaf cell}, sample jumping widths using $\overline{b}$ as upper bound and use these widths to select candidates within $S$.
Aggregating a subtree to a virtual leaf cell allows skipping leaf cells which do not contain candidates, but uses a weaker bound $\overline{b}$ and thus a potentially larger candidate set.
Thus, a fast algorithm requires an aggregation criterion which keeps both the number of candidates and the number of examined quadtree cells low.

As stated before, we record the number of points in each subtree during quadtree construction.
This information is now used for the query algorithm:
We aggregate a subtree $S$ to a virtual leaf cell exactly if $|S|$, the number of points contained in $S$, is below $1 / f(\text{dist}(S,q))$.
This corresponds to less than one expected candidate within $S$.
The changes required in Algorithm~\ref{algo:quadnode-probabilistic} to use the subtree aggregation are minor.
Lines~\ref{line:notleafcell}, \ref{line:candidate-confirmation} and \ref{line:add-confirmed} are changed to:

\LinesNotNumbered
\begin{algorithm}[H]
\nlset{5} \textbf{if}\emph{$c$ is inner node and $|c|\cdot\overline{b}\geq 1$}\textbf{ then}
\end{algorithm}

\begin{algorithm}[H]
 \nlset{14} neighbor = maybeGetKthElement($q$, $f$, $i$, $\overline{b}$, $c$)\;
 \nlset{15} add neighbor to $\neighborset$ if not null
\end{algorithm}

The main change consists in the use of the function maybeGetKthElement (Algorithm~\ref{algo:maybe-get-kth-element}, Appendix~\ref{sec:maybe-get-kth-element}).
Given a subtree $S$, an index $k$, $q$, $f$, and $\overline{b}$, this function descends $S$ to the leaf cell containing the
$k$th element. This element $p_k$ is then accepted with probability $f(\dist(q, p_k)) / \overline{b}$.

Since the upper bound calculated at the root of the aggregated subtree is not smaller than the individual upper bounds at the original leaf cells, Proposition~\ref{lemma:independent-correct-probabilities} also holds for the virtual leaf cells. This establishes the correctness.

The time complexity is given by the following theorem, whose proof can be found in Appendix~\ref{sec:proof-subtree-aggregation-complexitx}.

\begin{theorem}
Let $T$ be a quadtree with $n$ points and $(q,f)$ a query pair.
A query $(q,f)$ using subtree aggregation has time complexity $O((|\neighborset(q,f)| + \sqrt{n}) \log n)$ whp.
\label{lemma:subtree-aggregation-complexity}
\end{theorem}

\section{Application Case Studies}
\label{sec:applications}
In order to test our algorithm for PNQs, we apply it in two application case studies,
one for Euclidean, the other one for hyperbolic geometry. For the Euclidean case study
we build a simple disease spread simulation as an example for a probabilistic spreading process.
The probability distribution of points is in this case non-uniform and unknown. The hyperbolic application, in turn, is a 
generator for complex networks with a known point distribution. 

\subsection{Probabilistic Spreading}
\label{sub:disease-sim}
When both contact graph and travel patterns of a susceptible population are not known in detail, the resulting spreading behavior of an infectious disease seems probabilistic.
Contagious diseases usually spread to people in the vicinity of infected persons, but an infectious person occasionally bridges larger distances by travel and spreads the disease this way.
We model this effect with our probabilistic neighborhood function $f$, giving a higher probability for small distances and a lower but non-zero probability for larger distances.
Note that this scenario is meant as an example of the probabilistic spreading simulations possible with our algorithm and not as highly realistic from an epidemiological point of view.

In the simulation, the population is given as a set $P$ of points in the Euclidean plane.
 In the initial step, exactly one point (= person) from $P$ is marked as infected. 
Then, in each round, a PNQ is performed
for each infected person $q$.
All points in $N(q, f)$ become infected in the next round.
We use an SIR model~\cite{hethcote2000mathematics}, \ie previously infected persons recover with a 
certain probability in each round and stay infectious otherwise.
In our simulation, persons recover with a rate of 0.8 and are then immune.

\subsection{Random Hyperbolic Graph Generation}
\label{sec:application-hyperbolic-random-graphs}
Random hyperbolic graphs (RHGs, also see Section~\ref{sec:related}) are a generative graph model for
complex networks. For graph generation one places $n$ points (= vertices) randomly in a hyperbolic disk.
The radius $R$ of the disk can be used to control the average degree of the network.
A pair of vertices is connected by an edge with a probability that depends on the vertices' hyperbolic distance.
This connection probability is given in~\cite[Eq.~(41)]{Krioukov2010} and parametrized by a temperature $T\geq 0$:
\begin{equation}
 f(x) = \frac{1}{e^{(1/T)\cdot(x-R)/2}+1}
 \label{eq:Krioukov-equation}
\end{equation}
This definition of random hyperbolic graphs is a generalized version of the one considered in our previous work, which was restricted to the special case of $T=0$.

\subsection{Experimental Settings and Results}
\label{sub:exp-results}
Our implementation uses the NetworKit toolkit~\cite{staudt2014networkit} and is 
written in C++ 11. It is included in NetworKit release 4.1.
Running time measurements were made with g++ 4.8 -O3 on a machine with 128 GB RAM
and an Intel Xeon E5-1630 v3 CPU with four cores at 3.7 GHz base frequency. 
Our code is sequential, as is the reference implementation for random hyperbolic graph
generation~\cite{Aldecoa2015}.

\paragraph{Disease Spread Simulation.}
We experimented on three data sets
taken from NASA population density raster 
data~\cite{nasaGridPop} 
for Germany, France and the USA. They consist of rectangles with small square cells
(geographic areas) where for each cell the population from the year 2000 is given.
To obtain a set of points, we randomly distribute points in each cell to fit 1/20th of the population density.
Figure~\ref{plot:heatMap-and-hyperbolic-disk} (left) in the appendix shows an example with roughly 4 million points on the map of Germany.
The data sets of France and USA have roughly 3 and 14 million points, respectively.

The number of required queries naturally depends heavily on the simulated disease.
For our parameters, a number of 5000 queries is typically reached within the first dozen steps.
To evaluate the algorithmic speedup, Table~\ref{table:running-time-queries-country} compares running times for 5000 pairwise
distance probing (PDP) queries against 5000 fast PNQs 
on the three country datasets.
To obtain a similar total number of infections, we use a slightly different probabilistic neighborhood function for each country and divide by the population: $f(x) := (1/x)\cdot e^7/n$.
This results in a slower initial progression for the US.
Our algorithm achieves a speedup factor of at least two orders of magnitude, even including the quadtree construction time.

\begin{table}[bt]
\centering 
\begin{tabular}{l|l|l|l}
 Country & 5000 PDP queries & Construction QT & 5000 QT queries\\
 \hline
 France & 1007 seconds & 1.6 seconds & 1.2 seconds\\ 
 Germany & 1395 seconds & 2.8 seconds & 1.3 seconds\\
 USA & 4804 seconds & 8.7 seconds & 0.7 seconds
\end{tabular}
\smallskip 
\caption{Running time results for polar Euclidean quadtrees on population data.
The query points were selected uniformly at random from $\pointset$, the probabilistic neighborhood function is $f(x) := (1/x)\cdot e^7/n$.
}
\label{table:running-time-queries-country}
\end{table}

\paragraph{Random Hyperbolic Graph Generation.}
An example graph generated from hyperbolic geometry can be seen in Figure~\ref{plot:heatMap-and-hyperbolic-disk} (right) in the appendix.
We compare our generator using PNQs with the only (to our knowledge) previously existing generator for general random hyperbolic graphs~\cite{Aldecoa2015}, \ie those not only following the threshold model.
As seen in Figure~\ref{plot:time-scatter}, our implementation is faster by at least one order of magnitude and the experimental running times support our theoretical time complexity of $O((n^{3/2}+m)\log n)$.
A comparison of the generated graphs with those created by the existing implementation can be found in Appendix~\ref{sec:comparison-previous-impl}. The differences measured by a set of suitable network analysis metrics are within the
range of random fluctuations for the sample size of $80$.

\definecolor{markedcolor}{RGB}{31,120,180}
\definecolor{plottinggreen}{RGB}{178,223,138}
\definecolor{thirdhue}{RGB}{228,26,28}
\newcommand{\aconst}{5.605}
\newcommand{\bconst}{2.18}
\newcommand{\cconst}{1.77}
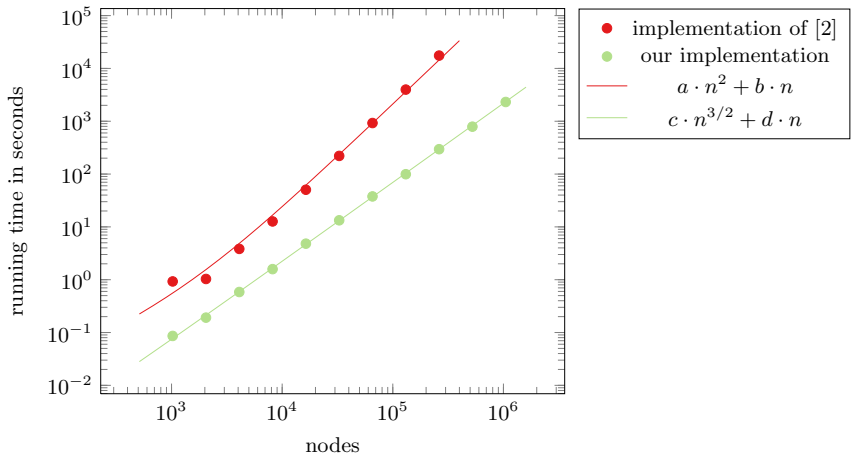
\begin{figure}[tb]
\centering

 \begin{tikzpicture}[scale=0.9]
 \begin{axis}[ymode=log,xmode=log,xlabel=nodes,ylabel=running time in seconds,legend entries={}, legend pos=outer north east]
 \addplot[only marks, thirdhue] table {plots/1465982453.5281825-cost-native.dat};\addlegendentry{implementation of \cite{Aldecoa2015}};
  \addplot[only marks, plottinggreen] table {plots/1465988533.119531-cost-PNQs-single-threaded.dat};\addlegendentry{our implementation};
  \addplot[thirdhue] expression[domain=512:400000] {2.089 * 10^(-7)*x^2 + 3.311*10^(-4)*x};\addlegendentry{$a\cdot n^2 + b\cdot n$};
  \addplot[plottinggreen] expression[domain=512:1600000] {\aconst *10^(-6)*x + \bconst * 10^(-6)*x*sqrt(x)};\addlegendentry{$c\cdot n^{3/2} + d\cdot n$};
 \end{axis}
\end{tikzpicture}

 \caption{Comparison of running times to generate networks with $2^{10}$-$2^{20}$ vertices, $\alpha=1$, $T=$0.5 and average degree $\overline{k}=6$. The gap between the running times widens, which in the loglog-plot implies a different exponent in the time complexities.
 Running times are fitted with $a=2.089 \cdot 10^{-7}$, $b=3.311\cdot10^{-4}$, $c=2.18\cdot10^{-6}$ and $d=5.6\cdot10^{-6}$. }
 \label{plot:time-scatter}
\end{figure}

\section{Conclusions}
After formally defining the notion of probabilistic neighborhoods, we have presented a
quadtree-based query algorithm for such neighborhoods in the Euclidean and hyperbolic plane.
Our analysis shows a time complexity of $O((|N(q,f)| + \sqrt{n})\log n)$, our algorithm is to the best of our knowledge the first to solve the problem asymptotically faster than pairwise distance probing.
With two example applications we have shown that our algorithm is also faster 
in practice by at least one order of magnitude.

\paragraph{Acknowledgements.}
This work is partially supported by German Research Foundation (DFG) grant ME 3619/3-1 within the 
Priority Programme 1736 \emph{Algorithms for Big Data}. The authors thank Mark Ortmann for helpful discussions.

\bibliographystyle{plain}
\bibliography{ProbabilisticQuadtree}

\clearpage
\appendix

\section{Related Algorithmic Idea}
\label{sec:fast-batagelji-brandes}
Our approach was inspired by the following algorithm with optimal linear running time
for Erd\H{o}s-R\'{e}nyi graph generation~\cite{batagelj2005efficient}.

\begin{algorithm}[H]
\KwIn{number of vertices $n$, edge probability $0 < p < 1$}
\KwOut{$G = (\{0,...,n-1\}, E) \in \mathcal{G}(n,p)$} 
 $E$ = $\emptyset$\;
 $v$ = 1\;
 $w$ = -1\;
 \While{$v < n$}{
 draw $r\in [0,1)$ uniformly at random\;
 $w = w + 1 + \lfloor \frac{\log (1-r)}{\log (1-p)} \rfloor$\;
 \While{$w \geq v$ and $v < n$}{
 $w = w-v$\;
 $v = v + 1$\;
 }
 \If{$v < n$}
 {add $\{u,v\}$ to $E$}
 } 
 \caption{Efficient neighborhood generation for Erd\H{o}s-R\'{e}nyi graphs~\cite{batagelj2005efficient}.}
 \label{algo:fast-batagelji-brandes}
\end{algorithm}

\section{Proofs of Section~\ref{sec:baseline}}
\label{sec:basic-qt-proofs}

\subsection{Proof of Lemma~\ref{lemma:node-cell-probabilities}}
\label{sub:proof-lemma-node-cell-probabilities}
\begin{proof}
Due to the similarity of Lemma~\ref{lemma:node-cell-probabilities} to \cite[Lemma 1]{Looz2015HRG}, the proof follows a similar structure.
Let $C$ be a quadtree cell at level $k$, delimited by $\textnormal{min}_r$, $\textnormal{max}_r$, $\textnormal{min}_\phi$ and $\textnormal{max}_\phi$.
As stated in Section~\ref{sub:notation}, we require the point probability distribution to be rotationally invariant.
The probability that a point $p$ is in $C$ is then given by 
\begin{equation}
 \Pr(p \in C) = \frac{\max_{\phi} - \min_{\phi}}{2\pi} \cdot (J(\textnormal{max}_r) - J(\textnormal{min}_r))\label{eq:cell-probability-mass}.
\end{equation}
The boundaries of the children of $C$ are given by the splitting rules in Section~\ref{sub:qt-augment}.
\begin{align}
 \mathrm{mid}_\phi &:= \frac{\max_\phi+\min_\phi}{2}\label{eq:angular-split} \\
 \mathrm{mid}_r &:= (J|_{[0,R]})^{-1}\left(\frac{J(\mathrm{max}_r) + J(\mathrm{min}_r)}{2}\right)\label{eq:radial-split}
\end{align}
We proceed with induction over the depth $i$ of $C$.
Start of induction ($i$ = 0):
At depth 0, only the root cell exists and covers the whole disk.
Since $C = \mathbb{D}_R$, $\Pr(p \in C) = 1 = 4^{-0}$.

Inductive step ($i \rightarrow i+1$):
Let $C_i$ be a node at depth $i$.
$C_i$ is delimited by the radial boundaries $\mathrm{min}_r$ and $\mathrm{max}_r$, as well as the angular boundaries $\mathrm{min}_\phi$ and $\mathrm{max}_\phi$.
It has four children at depth $i+1$, separated by $\mathrm{mid}_r$ and $\mathrm{mid}_\phi$. Let $SW$ be the south west child of $C_i$.
With Eq.~(\ref{eq:cell-probability-mass}), the probability of $p\in SW$ is:
\begin{equation}
\Pr(p \in SW) = \frac{\mathrm{mid}_{\phi} - \min_{\phi}}{2\pi} \cdot \left(J\left(\mathrm{mid}_r\right) - J\left(\textnormal{min}_r\right)\right)
 \end{equation}.

Using Equations~(\ref{eq:angular-split}) and~(\ref{eq:radial-split}), this results in a probability of 
\begin{align}
\Pr(p \in SW) &= \frac{\frac{\max_\phi+\min_\phi}{2} - \min_{\phi}}{2\pi} \cdot \left(J\left((J|_{[0,R]})^{-1}\left(\frac{J(\mathrm{max}_r) + J(\mathrm{min}_r)}{2}\right)\right) - J(\textnormal{min}_r)\right)\\
\Pr(p \in SW) &= \frac{\frac{\max_\phi+\min_\phi}{2} - \min_{\phi}}{2\pi} \cdot \left(\frac{J(\mathrm{max}_r) + J(\mathrm{min}_r)}{2} - J(\textnormal{min}_r)\right)\\
\Pr(p \in SW) &= \frac{\frac{\max_\phi-\min_\phi}{2}}{2\pi} \cdot \left(\frac{J(\mathrm{max}_r) - J(\mathrm{min}_r)}{2}\right)\\
\Pr(p \in SW) &= \frac{1}{4} \frac{\max_\phi-\min_\phi}{2\pi} \cdot \left(J(\mathrm{max}_r) - J(\mathrm{min}_r)\right)\\
\end{align}
As per the induction hypothesis, $\Pr(p \in C_i)$ is $4^{-i}$ and $\Pr(p \in SW)$ is thus $\frac{1}{4}\cdot 4^{-i} = 4^{-(i+1)}$.
Due to symmetry when selecting $\mathrm{mid}_\phi$, the same holds for the south east child of $C_i$. Together, they contain half of the probability mass of $C_i$.
Again due to symmetry, the same proof then holds for the northern children as well.
\qed
\end{proof}

\subsection{Proof of Lemma~\ref{lemma:bound-number-quadtree-cells}}
\label{sub:proof-lemma-bound-number-quadtree-cells}
\begin{proof}
A quadtree $T$ containing $n$ points can have at most $n$ non-empty leaf cells. We can thus bound the total number of leaf cells in $T$ by limiting the number of empty cells.

An empty leaf cell occurs when a previous leaf cell $c$ is split.
We consider two cases, depending on how many of the children of $c$ contain points:

\textbf{Case 1:} All but one of the children of $c$ are empty and all points in $c$ are concentrated in one child.
We call a split of this kind an \emph{excess} split, since it did not result in dividing the points in $c$.

\textbf{Case 2:} At least two children of $c$ contain points.

The number of excess splits caused by a pair of points depends on the area they are clustered in.
Two sufficiently close points could cause a potentially unbounded number of excess splits.
However, due to Lemma~\ref{lemma:node-cell-probabilities}, each child cell contains a quarter of the probability mass of its parent cell.
Given two points $p,q$ in a cell which is split, they end up in different child cells with probability 3/4.

The expected number of excess splits for a point $p$ is thus at most\footnote{Note that the real number of excess splits might be lower, since a split might separate another point from $p$ and $q$.}
\begin{equation}
\sum_{i=0}^{\infty} i\cdot 4^{-i} = \frac{4}{9}. 
\end{equation}

Due to the linearity of expectations, the expected number of excess splits caused by $n$ points is then at most $4n/9$.
Each excess split causes four additional quadtree nodes, three of them are empty leaf cells.

If we remove all quadtree nodes caused by excess splits and reconnect the tree by connecting the remaining leaves
to their lowest unremoved ancestor, every inner node in the remaining tree $T'$ has at least two non-empty subtrees.
Since a binary tree with $n$ leaves has $O(n)$ inner nodes~\cite{Samet:2005:FMM:1076819} and the branching factor in $T'$ is at least two, $T'$ also contains at most $O(n)$ inner nodes.

Together with the expected $O(n)$ nodes caused by excess splits, this results in $O(n)$ nodes in $T$ in expectation.
\qed
\end{proof}

\subsection{Proof of Proposition~\ref{thm:quadtree-height}}
\label{sub:proof-quadtree-height}
\begin{proof}
We proved a similar lemma in previous work~\cite{Looz2015HRG}, for hyperbolic geometry only and a restricted family of probability distributions.
The requirement for that proof was that a given point $p$ has a probability of $4^{-i}$ to land in a given cell at depth $i$.
In Lemma~\ref{lemma:node-cell-probabilities}, we show that this requirement is fulfilled for the quadtrees used in this paper in both Euclidean and hyperbolic geometry.
We can thus reuse the proof of \cite[Lemma 2]{Looz2015HRG}, which we include for the purpose of self-containment:
\end{proof}

\subsubsection{Proof of \cite[Lemma 2]{Looz2015HRG}}
\newcommand{\variable}{n}
\newcommand{\binvariable}{n}

\begin{proof} 
In a complete quadtree, $4^i$ cells exist at depth $i$. For analysis purposes only, we construct such 
a complete but initially empty quadtree of height $k = 3\cdot\lceil \log_4(n)\rceil$, which has at least $n^3$ leaf cells.
As seen in Lemma~\ref{lemma:node-cell-probabilities}, a given point has an equal chance to land in each leaf cell.
Hence, we can apply \cite[Lemma~6]{Looz2015HRG} with each leaf cell being a bin and a point being a ball.
(The fact that we can have more than $n^3$ leaf cells only helps in reducing the average load.)
From this we can conclude that, for $n$ sufficiently large, no leaf cell of the current tree contains more than 1 point with high probability (whp).
Consequently, the total quadtree height does not exceed $k = 3\cdot\lceil \log_4(n)\rceil \in O(\log n)$ whp.

Let $T'$ be the quadtree as constructed in the previous paragraph, starting with a complete quadtree of height~$k$ and splitting leaves when their capacity is exceeded.
Let $T$ be the quadtree created in our algorithm, starting with a root node, inserting points and also splitting leaves when necessary, growing the tree downward.

Since both trees grow downward as necessary to accommodate all points, but $T$ does not start with a complete quadtree of height~$k$, the set of quadtree nodes in $T$ is a subset of the quadtree nodes in $T'$.
Consequently, the height of $T$ is bounded by $O(\log n)$ whp as well.
\qed	
\end{proof}

\subsection{Proof of Proposition~\ref{lemma:independent-correct-probabilities}}
\label{sub:proof-lemma-indep}
\begin{proof}
Note that the hyperbolic [Euclidean] distances, which are mapped to probabilities according to the function $f$,
are calculated by Algorithm~\ref{algo:hyperbolic-distances}
[Algorithm~\ref{algo:Euclidean-polar-distances}], which are presented in Appendix~\ref{sec:distance} (together with their correctness proofs).
We continue the current proof with details for all three main steps.
\paragraph{Step 1:}
Between two points, the jumping width $\delta$ is given by Line~\ref{line:deltaskip}.
The probability that exactly $i$ points are skipped between two given candidates is $(1-\overline{b})^i\cdot \overline{b}$:
\begin{align}
 \Pr(i \leq \delta < i+1)&=\\
 \Pr(i \leq \ln(1-r)/\ln(1-\overline{b})<i+1)&=\label{eq:logratio}\\
 \Pr(\ln(1-r)\leq i\cdot\ln(1-\overline{b}) \wedge \ln(1-r)>(i+1)\cdot\ln(1-\overline{b}))&=\label{eq:loginequality}\\
 \Pr(1-(1-\overline{b})^i \leq r < 1-(1-\overline{b})^{i+1})&=\label{eq:last-probability}\\
 1-(1-\overline{b})^{i+1} - 1+(1-\overline{b})^i&=\label{eq:uniform-r-needed}\\
 (1-\overline{b})^i(1-(1-\overline{b})) &= \\
 (1-\overline{b})^i\cdot \overline{b}\label{eq:waiting-times}
\end{align}
Note that in Eq.~(\ref{eq:logratio}) the denominator is negative, thus the direction of the inequality is reversed in the transformation.
The transformation from Eq.~(\ref{eq:last-probability}) to Eq.~(\ref{eq:uniform-r-needed}) works since $r$ is uniformly distributed.

Following from Eq.~(\ref{eq:waiting-times}), the probability is $\overline{b}$ for $i=0$, and if a point is selected as a candidate, the subsequent point is selected with a probability of $\overline{b}$.

\paragraph{Step 2:}
Let $p_i$, $p_j$ and $p_l$ be points in a leaf, with $i<j<l$ and let $p_i$ be a neighbor candidate.
For now we assume that no other points in the same leaf are candidates and consider the probability that $p_l$ is selected as a candidate depending on whether the intermediate point $p_j$ is a candidate.

\textbf{Case 2.1:} If point $p_j$ is a candidate, then point $p_l$ is selected if $l-j$ points are skipped after selecting $p_j$.
Due to Step 1, this probability is $(1-\overline{b})^{l-j}\cdot \overline{b}$

\textbf{Case 2.2:} If point $p_j$ is \emph{not} a candidate, then point $p_l$ is selected if $l-i$ points are skipped after selecting $p_i$.
Given that $p_j$ is not selected, at least $j-i$ points are skipped.
The conditional probability is then:
\begin{align}
 \Pr(l-i \leq \delta < l-i+1 | \delta > j-i) &=\\
 \Pr(1-(1-\overline{b})^{l-i} < r < (1-(1-\overline{b})^{l-i+1}) | \delta > j-i)&=\\
 (1-\overline{b})^{l-i}\cdot \overline{b} / (1-\overline{b})^{j-i}&=\\
 (1-\overline{b})^{l-j}\cdot \overline{b}
\end{align}
As both cases yield the same result, the probability $\Pr(p_l \in \candidateset)$ is independent of whether $p_j$ is a candidate.

\paragraph{Step 3:}
Let $C$ be a leaf cell in which all points up to point $p_i$ are selected as candidates.
Due to Step 1, the probability that $p_{i+1} $ is also a candidate, meaning no points are skipped, is $(1-\overline{b})^0\cdot \overline{b} = \overline{b}$.
Due to Step 2, the probability of $p_{i+1}$ being a candidate is independent of whether $p_i$ is a candidate.
This can be applied iteratively until the beginning of the leaf cell, yielding a probability of $\overline{b}$ for $p_i$ being a candidate, independent of whether other points are selected.

A neighbor candidate $p_i$ is accepted as a neighbor with probability $f(\dist(p_i, q))/\overline{b}$ in Line~\ref{line:candidate-confirmation}.
Since $\overline{b}$ is an upper bound for the neighborhood probability, the acceptance ratio is between 0 and 1.
The probability for a point $p$ to be in the probabilistic neighborhood computed by Algorithm~\ref{algo:quadnode-probabilistic} is thus:
\begin{align}
 \Pr(p\in \neighborset(q,f)) &=\\
 \Pr(p\in \neighborset(q,f) \wedge p\in \candidateset(q,f)) &=\\
 \Pr(p \in \neighborset(q,f) | p \in \candidateset(q,f)) \cdot \Pr(p\in \candidateset(q,f))&=\\
 f(\dist(p, q))/\overline{b} \cdot \overline{b}&=\\
 f(\dist(p, q))
\end{align}
\qed
\end{proof}

\subsection{Proof of Proposition~\ref{cor:basic-time}}
\label{sec:proof-cor-basic-time}
\begin{proof}
The total time complexity of the query algorithm is determined by the number of recursive calls (Line~\ref{line:calling-child}) and the number of loop iterations (Line~\ref{line:loop-iteration}).
During tree traversal, one recursive call is made for each examined quadtree node.
During examination of a leaf, one loop iteration happens for every examined candidate.
Let the set of neighbors ($\neighborset(q,f)$), candidates ($\candidateset(q,f)$) and examined cells ($\cellset(q, f)$) be as defined in Section~\ref{sub:notation}.
The time complexity of the query is then in $\Theta(|\candidateset(q,f)| + |\cellset(q,f)|)$.

All cells of the quadtree are examined, thus $\cellset(q,f) = \cellset(T)$.
If the cells are split using the medians of point positions, then no leaf cell is empty and the tree contains at most $n$ cells.
If cells are split using the theoretical probability distributions, the tree contains at most $O(n)$ cells in expectation due to Lemma~\ref{lemma:bound-number-quadtree-cells}.
It follows that the number of examined cells is in $\Theta(n)$ in expectation.
Since the candidate set is a subset of the point set, the expected number of candidates is at most $n$.
The query time complexity is then in $O(n) +  \Theta(|\cellset(T)|$ = $\Theta(n)$ in expectation.
\qed
\end{proof}

\newcommand{\lE}{\ensuremath{\mathrm{leftExtremum}}}
\newcommand{\rE}{\ensuremath{\mathrm{rightExtremum}}}


\subsection{Distance between Quadtree Cell and Point}
\label{sec:distance}
To calculate the upper bound $\overline{b}$ used in Algorithm~\ref{algo:quadnode-probabilistic}, we need a lower bound for the distance between the query point $q$ and any point in a given quadtree cell.
Since the quadtree cells are polar, the distance calculations might be unfamiliar and we show and prove them explicitly.
For the hyperbolic case, the distance calculations are shown in Algorithm~\ref{algo:hyperbolic-distances} and proven in Lemma~\ref{lemma:hyperbolic-distances}.
The Euclidean calculations are shown in Algorithm~\ref{algo:Euclidean-polar-distances} and proven in Lemma~\ref{lemma:Euclidean-polar-distances}.

\LinesNumbered
\begin{algorithm}
\KwIn{quadtree cell $C$ = ($\text{min}_r$, $\text{max}_r$, $\text{min}_\phi$, $\text{max}_\phi$), query point $q = (\phi_q, r_q)$}
\KwOut{infimum and supremum of hyperbolic distances $q$ to interior of $C$}
\tcc{start with corners of cell as possible extrema}
cornerSet = \{($\text{min}_\phi$, $\text{min}_r$), ($\text{min}_\phi$, $\text{max}_r$), ($\text{max}_\phi$, $\text{min}_r$), ($\text{max}_\phi$, $\text{max}_r$)\}\;
a = $\cosh(r_q)$\;
b = $\sinh{r_q}\cdot\cos(\phi_q-\text{min}_\phi)$\;
\tcc{Left/Right boundaries}
leftExtremum = $\frac{1}{2}\ln\left(\frac{a+b}{a-b}\right)$\label{line:left-extremum}\;
\If{$\text{min}_r < \lE < \text{max}_r$}{
add $(\text{min}_\phi, \lE)$ to cornerSet\;
}

b = $\sinh{r_q}\cdot\cos(\phi_q-\text{max}_\phi)$\;
rightExtremum = $\frac{1}{2}\ln\left(\frac{a+b}{a-b}\right)$\label{line:right-extremum}\;
\tcc{Top/bottom boundaries}
\If{$\text{min}_r < \rE < \text{max}_r$}{
add $(\text{max}_\phi, \rE)$ to cornerSet\;
}

\If{$\text{min}_\phi < \phi_q \text{max}_\phi$}{
add $(\phi_q, \text{min}_r)$ and $(\phi_q, \text{max}_r)$ to cornerSet\;
}
$\phi_{\text{mirrored}} = \phi_q + \pi \mod 2\pi$\;

\If{$\text{min}_\phi < \phi_{\text{mirrored}} < \text{max}_\phi$}{
add $(\phi_{\text{mirrored}}, \text{min}_r)$ and $(\phi_{\text{mirrored}}, \text{max}_r)$ to cornerSet\;
}
\tcc{If point is in cell, distance is zero:}
\If{$\text{min}_\phi \leq \phi_q < \text{max}_\phi \text{ AND }\text{min}_r \leq r_q < \text{max}_r$}{
infimum = 0\;
}
\Else{
infimum = $\min_{e \in \text{cornerSet}} \text{dist}_\hyperbolic{} (q, e)$\;
}
supremum = $\max_{e \in \text{cornerSet}} \text{dist}_\hyperbolic{} (q, e)$\;
\Return{infimum, supremum};

 \caption{Infimum and supremum of distance in a hyperbolic polar quadtree}
 \label{algo:hyperbolic-distances}
\end{algorithm}

\begin{algorithm}
\KwIn{quadtree cell $C$ = ($\text{min}_r$, $\text{max}_r$, $\text{min}_\phi$, $\text{max}_\phi$), query point $q = (\phi_q, r_q)$}
\KwOut{infimum and supremum of Euclidean distances $q$ to interior of $C$}
\tcc{start with corners of cell as possible extrema}
 cornerSet = \{($\text{min}_\phi$, $\text{min}_r$), ($\text{min}_\phi$, $\text{max}_r$), ($\text{max}_\phi$, $\text{min}_r$), ($\text{max}_\phi$, $\text{max}_r$)\}\;
\tcc{Left/Right boundaries}
\lE = $r_q\cdot\cos(\text{min}_\phi - \phi_q)$\;\label{line:left-extremum-Euclidean}
\If{$\text{min}_r < \lE < \text{max}_r$}{
add $(\text{min}_\phi, \lE)$ to cornerSet\;
}

\rE = $r_q\cdot\cos(\text{max}_\phi - \phi_q)$\;\label{line:right-extremum-Euclidean}
\If{$\text{min}_r < \rE < \text{max}_r$}{
add $(\text{max}_\phi, \rE)$ to cornerSet\;
}

\tcc{Top/bottom boundaries}
\If{$\text{min}_\phi < \phi_q < \text{max}_\phi$}{
add $(\phi_q, \text{min}_r)$ and $(\phi_q, \text{max}_r)$ to cornerSet\;
}
$\phi_{\text{mirrored}} = \phi_q + \pi \mod 2\pi$\;

\If{$\text{min}_\phi < \phi_{\text{mirrored}} < \text{max}_\phi$}{
add $(\phi_{\text{mirrored}}, \text{min}_r)$ and $(\phi_{\text{mirrored}}, \text{max}_r)$ to cornerSet\;
}
\tcc{If point is in cell, distance is zero:}
\If{$\text{min}_\phi \leq \phi_q < \text{max}_\phi \text{ AND }\text{min}_r \leq r_q < \text{max}_r$}{
infimum = 0\;
}
\Else{
infimum = $\min_{e \in \text{cornerSet}} \text{dist}_\hyperbolic{} (q, e)$\;
}
supremum = $\max_{e \in \text{cornerSet}} \text{dist}_\hyperbolic{} (q, e)$\;
\Return{infimum, supremum};
 \caption{Infimum and supremum of distance in a Euclidean polar quadtree}
 \label{algo:Euclidean-polar-distances}
\end{algorithm}

\pagebreak 

\begin{lemma}
 Let $C$ be a quadtree cell and $q$ a point in hyperbolic space.
 The first value returned by Algorithm~\ref{algo:hyperbolic-distances} is the distance of $C$ to $q$.
 \label{lemma:hyperbolic-distances}
\end{lemma}

\begin{proof}
When $q$ is in $C$, the distance is trivially zero.
Otherwise, the distance between $q$ and $C$ can be reduced to the distance between $q$ and the boundary of $C$, $\delta C$:
\begin{equation}
\text{dist}_\hyperbolic{} (C, q) = \text{dist}_\hyperbolic{} (\delta C, q) = \inf_{p \in \delta C} \text{dist}_\hyperbolic{} (p, q) 
\end{equation}
Since the boundary is closed, this infimum is actually a minimum:
\begin{equation}
\text{dist}_\hyperbolic{} (C, q) = \inf_{p \in \delta C} \text{dist}_\hyperbolic{} (p, q) = \min_{p \in \delta C} \text{dist}_\hyperbolic{} (p, q) 
\end{equation}
The boundary of a quadtree cell consists of four closed curves:
\begin{itemize}
 \item left: $\{(\text{min}_\phi, r) |  \text{min}_r \leq r \leq \text{max}_r \} $
 \item right: $\{(\text{max}_\phi, r) |  \text{min}_r \leq r \leq \text{max}_r \} $
 \item lower: $\{(\phi, \text{min}_r) | \text{min}_\phi \leq \phi \leq \text{max}_\phi \} $
 \item upper: $\{(\phi, \text{max}_r) | \text{min}_\phi \leq \phi \leq \text{max}_\phi \} $
\end{itemize}
We write the distance to the whole boundary as a minimum over the distances to its parts:
\begin{equation}
 \text{dist}_\hyperbolic{} (\delta C, q) = \min_{A \in \{\text{left, right, lower, upper} \}} \text{dist}_\hyperbolic{} (A, q)
\end{equation}

All points on an angular boundary curve $A$ have the same angular coordinate $\phi_A$.
Let $d_A(r) = \mathrm{acosh}(\cosh(r)\cosh(r_q) - \sinh(r)\sinh(r_q) \cos(\phi_q - \phi_A))$ for a fixed point $q$.
The distance $\text{dist}_\hyperbolic{} (A, q)$ can then be reduced to:
\begin{align}
 \text{dist}_\hyperbolic{} (A, q) &= \min_{\text{min}_r \leq r \leq \text{max}_r} d_A(r)\\
\end{align}
The minimum of $d_A$ on $A$ is the minimum of $d_A(\text{min}_r),$  $d_A(\text{max}_r)$ and the value at possible extrema.
To find the extrema, we define a function $g(r) = \cosh(r)\cosh(r_q) - \sinh(r)\sinh(r_q) \cos(\phi_q - \phi_A)$.
Since $\mathrm{acosh}$ is strictly monotone, $g(r)$ has the same extrema as $d_A(r)$.

The factors $\cosh(r_q)$ and $\sinh(r_q) \cos(\phi_q - \phi_A)$ do not depend on $r$, to increase readability we substitute them with the constants $a$ and $b$:
\begin{align}
 a &= \cosh(r_q)\\
 b &= \sinh(r_q) \cos(\phi_q - \phi_A)\\
 d_A(r) &= \mathrm{acosh}(\cosh(r)\cdot a - \sinh(r)\cdot b)\\
 g(r) &= \cosh(r)\cdot a - \sinh(r)\cdot b
\end{align}
The derivative of $g$ is thus:
\begin{equation}
g'(r) = \sinh(r)\cdot a - \cosh(r)\cdot b = \frac{e^r-e^{-r}}{2}\cdot a - \frac{e^r+e^{-r}}{2}\cdot b
\end{equation}
With some transformations, we get the roots of $g'(r)$:
\paragraph*{Case $a=b$:}
 \begin{align}
  g'(r) &= 0 \Leftrightarrow\\
  \frac{e^r-e^{-r}}{2}\cdot a &= \frac{e^r+e^{-r}}{2}\cdot a\\
  e^r-e^{-r} &= e^r+e^{-r}\\
  -e^{-r} &= e^{-r}\\
  e^{-r} &= 0\\
 \end{align}
For $a=b$, $d_A$ has no extrema in $\mathbb{R}$.

\paragraph*{ $a\not=b$:}

\begin{align}
g'(r) &= 0 \Leftrightarrow\\
\frac{e^r-e^{-r}}{2}\cdot a &= \frac{e^r+e^{-r}}{2}\cdot b\Leftrightarrow\\
a e^r-ae^{-r} &=  be^r+be^{-r}\Leftrightarrow\\
(a-b)e^r - (a+b)e^{-r} &= 0\Leftrightarrow\\
(a-b)e^r &= (a+b)e^{-r}\Leftrightarrow\\
e^r &= \frac{a+b}{a-b}e^{-r}\Leftrightarrow\\
e^{2r} &= \frac{a+b}{a-b}\Leftrightarrow\\
2r &= \ln\left(\frac{a+b}{a-b}\right)\Leftrightarrow\\
r &= \frac{1}{2}\ln\left(\frac{a+b}{a-b}\right)\label{eq:left-right-extremum}
\end{align}
For  $a\not=b$, $d_A$ has a single extremum at $\frac{1}{2}\ln\left(\frac{a+b}{a-b}\right)$.
This extremum is calculated for both angular boundaries in Lines \ref{line:left-extremum} and \ref{line:right-extremum} of Algorithm~\ref{algo:hyperbolic-distances}.

If $d(r)$ has an extremum $x$ in $A$, the minimum of $d_A(r)$ on $A$ is $\min \{d_A(\text{min}_r)$,  $d_A(\text{max}_r)$, $d_A(x)\}$,
otherwise it is $\min \{d_A(\text{min}_r)$,  $d_A(\text{max}_r)\}$.

\vspace{1\baselineskip}
A similar approach works for the radial boundary curves. Let $B$ be a radial boundary curve at radius $r_B$ and angular bounds $\text{min}_\phi$ and $\text{max}_\phi$.
Let $d_B(\phi)$ be the distance to $q$ restricted to radius $r_B$.
\begin{align}
d_B &: [0,2\pi] \rightarrow \mathbb{R}\\
 d_B(\phi) &= \mathrm{acosh}(\cosh(r_B)\cosh(r_q) - \sinh(r_B)\sinh(r_q) \cos(\phi_q - \phi))
\end{align}
Similarly to the angular boundaries, we define some constants and a function $g(\phi)$ with the same extrema as $d_B$:
\begin{align}
  a &= \cosh(r_B)\cosh(r_q)\\
  b &= \sinh(r_B)\sinh(r_q)\\
 g(\phi) &= a - b \cos(\phi_q - \phi)
\end{align}

\paragraph*{Case: $b = 0$:}
\begin{align}
 b &= \sinh(r_B)\sinh(r_q) = 0 \Leftrightarrow\\
 g(\phi) &= a 
\end{align}
Since $g$ is constant, no extrema exist.

\paragraph*{Case: $b \not= 0$:}
We obtain the extrema with some transformations:
\begin{align}
 g'(\phi) &= -b \sin(\phi_q - \phi)\\
 g'(\phi) &= 0 \Leftrightarrow\\
 \sin(\phi_q - \phi) &= 0 \Leftrightarrow\\
 \phi &= \phi_q \mod \pi
\end{align}
The distance function $d_B(\phi)$ thus has two extrema.

The minimum of $d_B(r)$ on $B$ is then:
\begin{equation}
\min_{r \in B} d_B(r) = \min \{d_B(\text{min}_r), d_B(\text{max}_r)\} \cup \{d_B(\phi) |  \text{min}_\phi \leq \phi \leq \text{max}_\phi \wedge \phi = \phi_q \mod \pi \}
\end{equation}

The distance $\text{dist}_\hyperbolic{} (C, q)$ can thus be written as the minimum of four to ten point-to-point distances. 
Algorithm~\ref{algo:hyperbolic-distances} collects the arguments for these distances in the variable cornerSet and returns the distance minimum as the first return value.
\qed
\end{proof}

\begin{lemma}
 Let $T$ be a polar quadtree in Euclidean space, $c$ a quadtree cell of $T$ and $q$ a point in Euclidean space.
 The first value returned by Algorithm~\ref{algo:Euclidean-polar-distances} is the distance of $c$ to $q$.
 \label{lemma:Euclidean-polar-distances}
\end{lemma}

\begin{proof}
The general distance equation for polar coordinates in Euclidean space is
\begin{equation}
 f(r_p, r_q, \phi_p, \phi_q) = \sqrt{r_p^2 + r_q^2 -2 r_p r_q \cos(\phi_p - \phi_q)}
 \label{eq:distance-equation}
\end{equation}

If the query point $q$ is within $C$, the distance is zero.
Otherwise, the distance between $q$ and $C$ is equal to the distance between $q$ and the boundary of $C$.
We consider each boundary component separately and derive the extrema of the distance function.

\paragraph{Radial boundary.}
When considering the radial boundary, everything but one angle is fixed:
\begin{equation}
  f(\phi_p) = \sqrt{r_p^2 + r_q^2 -2 r_p r_q \cos(\phi_p - \phi_q)}
  \label{eq:distance-only-angle}
\end{equation}
Since the distance is positive and the square root is a monotone function, the extrema of the previous function are at the same values as the extrema of its square $g(\phi)$:
\begin{equation}
  g(\phi_p) =r_p^2 + r_q^2 -2 r_p r_q \cos(\phi_p - \phi_q)
  \label{eq:distance-only-angle-squared}
\end{equation}
We set the derivative to zero to find the extrema:
\begin{align}
  g'(\phi_p) &= 0 \Leftrightarrow\\
  2 r_p r_q \sin(\phi_p - \phi_q)\cdot (\phi_p - \phi_q) &= 0\\
  \phi_p = \phi_q \mod \pi
  \label{eq:distance-only-angle-squared-derivative}
\end{align}

\paragraph{Angular boundary.}
Similar to the radial boundary, we fix everything but the radius:
\begin{equation}
  f(r_p) = \sqrt{r_p^2 + r_q^2 -2 r_p r_q \cos(\phi_p - \phi_q)}
  \label{eq:distance-only-radius}
\end{equation}

Again, we define a helper function with the same extrema:
\begin{equation}
  g(r_p) =r_p^2 + r_q^2 -2 r_p r_q \cos(\phi_p - \phi_q)
  \label{eq:distance-only-radius-squared}
\end{equation}
We set the derivative to zero to find the extrema:
\begin{align}
  g'(r_p) &= 0 \Leftrightarrow\\
  2r_p - 2r_q\cos(\phi_p - \phi_q) &= 0\Leftrightarrow\\
  r_p &= r_q\cos(\phi_p - \phi_q)\Rightarrow\\
  g(r_p) &= r_p^2 + r_q^2 -2 r_p^2\\ &= r_q^2 - r_p^2\\ &= r_q^2(1-\cos(\phi_p-\phi_q))
  \label{eq:distance-only-angle-radius-derivative}
\end{align}

An extremum of $f$ on the boundary of cell $c$ is either at one of its corners or at the points derived in Eq.~(\ref{eq:distance-only-angle-squared-derivative}) or Eq.~(\ref{eq:distance-only-angle-radius-derivative}).
If $q\not\in c$, the minimum over these points and the corners, as computed by Algorithm~\ref{algo:Euclidean-polar-distances}, is the minimal distance between $q$ and any point in $c$.
If $q$ is contained in $c$, the distance is trivially zero.
\qed
\end{proof}

\pagebreak

\section{Algorithm maybeGetKthElement, used in Section~\ref{sec:subtree-aggr}}
\label{sec:maybe-get-kth-element}
\begin{algorithm}[H]
\KwIn{query point $q$, function $f$, index $k$, bound $\overline{b}$, subtree $S$}
\KwOut{$k$th point of $S$ or empty set}
\If{$S$.isLeaf()}{
acceptance = $f(\text{dist}(q,\mathrm{S.points}[k]))/\overline{b} $\;
\If{$\mathit{1-rand()} <$  acceptance}{
\Return $\mathrm{S.points}[k]$\;
}
\Else{
\Return $\emptyset$\;
}
}
\Else{\tcc{Recursive call}
offset := 0\;
\For{child $\in$ $S$.children}{
\If{$k - \mathrm{offset} < |\mathrm{child}|$}{\tcc{|child| is the number of points in \emph{child}}
\Return maybeGetKthElement($q$, $f$, $k$ - offset, $\overline{b}$, child)\;
}
offset += $|\mathrm{child}|$\;
}
}
 \caption{maybeGetKthElement}
 \label{algo:maybe-get-kth-element}
\end{algorithm}
\clearpage


\section{Proof of Theorem~\ref{lemma:subtree-aggregation-complexity}}
\label{sec:proof-subtree-aggregation-complexitx}
\begin{proof}
Similar to the baseline algorithm, the complexity of the faster query is determined by the number of recursive calls and the total number of loop iterations across the calls.
The first corresponds to the number of examined quadtree cells, the second to the total number of candidates.
With subtree aggregation, we obtain improved bounds: Lemma~\ref{lemma:subtree-aggregation-candidates} limits the number of candidates to $O(|\neighborset(q,f)| + \sqrt{n})$ whp, while Lemma~\ref{lemma:subtree-aggregation-cells} bounds the number of examined quadtree cells to $O((|\neighborset(q,f)| + \sqrt{n})\log n)$ whp.
Together, this results in a query complexity of $O((|\neighborset(q,f)| + \sqrt{n}) \log n)$ whp.
\qed
\end{proof}

\label{subsubsec:subtree-complexity-notation}
For the lemmas required in the proof of Theorem~\ref{lemma:subtree-aggregation-complexity} we need to introduce some
notation:
Let $T$ be a quadtree with $n$ points, $S$ a subtree of $T$ containing $s$ points, $q$ a query point and $f$ a function mapping distances to probabilities.
The set of neighbors ($\neighborset(q,f)$), candidates ($\candidateset(q,f)$) and examined cells ($\cellset(q, f)$) are defined as in Section~\ref{sub:notation}.

For the analysis we divide the space around the query point $q$ into infinitely many bands, based on the probabilities given by $f$.
A point $p\in\pointset$ is in band $i$ exactly if the probability of it being a neighbor of $q$ is between $2^{-(i+1)}$ and $2^{-i}$:
\[
 p \in \text{band }i  \Leftrightarrow 2^{-(i+1)} < f(\text{dist}(p,q)) \leq 2^{-i}
\]
Based on these bands, we divide the previous sets into infinitely many subsets:
\begin{itemize}
 \item $\pointset(q, f, i) := \{v \in \pointset | 2^{-(i+1)} < f(\text{dist}(v,q)) \leq 2^{-i}\}$
 \item $\neighborset(q, f, i) := \neighborset(q, f) \cap \pointset(q,f,i)$
 \item $\candidateset(q, f, i) := \candidateset(q, f) \cap \pointset(q,f,i)$
 \item $\cellset(q, f, i) := \{c \in \cellset(q, f) | 2^{-(i+1)} < f(\text{dist}(c,q)) \leq 2^{-i}\}$
\end{itemize}

Note that for fixed $n$, all but at most finitely many of these sets are empty.
We call the quadtree cells in $\cellset(q,f,i)$ to be \emph{anchored} in band $i$.
The region covered by a quadtree cell is in general not aligned with the probability bands, thus a quadtree cell anchored in band $i$ ($c \in \cellset(q,f,i)$) may contain points from higher bands (i.e. with lower probabilities).

We continue with two auxiliary results used in Lemma~\ref{lemma:subtree-aggregation-candidates}:
Lemma~\ref{lemma:half-prob-set} helps in bounding the number of candidates that are in the same band as their (virtual or original) quadtree cell is anchored in.
Lemma~\ref{lemma:ring-root} is used to bound the number of points in a higher band than their original quadtree cell.

\newcommand{\uaqcset}{\ensuremath{\Upsilon}}
\newcommand{\subtreeset}{\ensuremath{\mathcal{\Psi}}\xspace}
\begin{lemma}
Let $n$ be a natural number and let $A$, $B$ be sets with $A \subseteq B, |B| \leq n$ and the following property: $\Pr(b \in A) \geq 0.5$, $\forall b \in B$.
Further, let the probabilities for membership in $A$ be independent.
Then, the number of points in $B$ is in $O(|A| + \log n)$ with probability at least $1-1/n^3$.
\label{lemma:half-prob-set}
\end{lemma}
\begin{proof}
Let $X = |A|$ be a random variable denoting the size of $A$.
Since the individual probabilities for membership in $A$ might be different, $X$ does not necessarily follow a binomial distribution.
We define an auxiliary distribution $Y := B(|B|, 0.5)$.
Since all membership probabilities for $A$ are at least 0.5, lower tail bounds derived for $Y$ also hold for $X$.

The probability that $Y$ is less than $0.1|B|$ is then~\cite{Hoe63}:
\begin{align}
\Pr(Y < 0.1|B|) &\leq \exp\left(-2\frac{(0.5|B|-0.1|B|)^2}{|B|} \right)\\
	&= \exp\left(-2\frac{(0.4|B|)^2}{|B|}\right)\\
	&= \exp\left(-2 \cdot 0.16|B|\right)\\
	&= \exp\left(-0.32|B|\right)\\ 
\end{align}

Similar to the proof of Lemma~\ref{lemma:subtree-aggregation-cells}, we conclude with a case distinction:
\paragraph{If $|B| > 10\log n$:}
The probability $\Pr(|A| < 0.1|B|)$ is then $\Pr(|A| < 0.1|B|) \leq \Pr(Y < 0.1|B|) \leq \exp\left(-3.2\log n\right) = n^{-3.2} < 1/n^3$.
Thus $|B| \leq 10|A| \in O(|A|)$ with probability at least $1-1/n^3$.

\paragraph{If $|B| < 10\log n$:}
$|B|$ is then trivially in $O(\log n)$.
\qed
\end{proof}

\begin{lemma}
Let $T$ be a polar hyperbolic [Euclidean] quadtree with $n$ points and $s < n$ a natural number.
Let $\Lambda$ be a circle in the hyperbolic [Euclidean] plane and let \subtreeset be the disjoint set of subtrees of $T$ that contain at most $s$ points and are cut by $\Lambda$.
Then, the subtrees in \subtreeset contain at most $24\sqrt{n\cdot s}$ points with probability at least $1-0.7^{\sqrt{n}}$ for $n$ sufficiently large.
\label{lemma:ring-root}
\end{lemma}

%
\begin{proof}
This proof is adapted from \cite[Lemma 3]{Looz2015HRG}.
Let $k := \lfloor \log_4 n/s\rfloor$ be the minimal depth at which cells have at least $s$ points in expectation.
At most $4^k$ cells exist at depth $k$, defined by at most $2^{k}$ angular and $2^{k}$ radial divisions.
When following the circumference of the query circle $\Lambda$, each newly cut cell requires the crossing of an angular or radial division.
Each radial and angular coordinate occurs at most twice on the circle boundary, thus each division can be crossed at most twice.
With two types of divisions, $\Lambda$ crosses at most $2\cdot2\cdot 2^k = 4\cdot2^{\lfloor\log_4 n/s \rfloor}$ cells at depth $k$.
Since the value of $4\cdot2^{\lfloor\log_4 n/s \rfloor}$ is at most $4\cdot2^{\log_4 n/s}$, this yields $\leq 8\cdot \sqrt{n/s}$ cut cells.
We denote the set of cut cells with \ringset.
Since the cells in \ringset cover the circumference of the circle $\Lambda$, a subtree $S$ which is cut by $\Lambda$ is either contained within one of the cells in \ringset,
corresponds to one of the cells or contains one.
In the first two cases, all points in $S$ are within the cells of \ringset.
In the second case, at least one cell of \ringset is contained in $S$.
As the subtrees are disjoint, this cell cannot be contained in any other of the considered subtrees.
Thus, there are no more subtrees containing points not in \ringset than there are cells in \ringset, which are less than $8\cdot \sqrt{n/s}$ many.

Due to Lemma~\ref{lemma:node-cell-probabilities}, the probability that a given point is in a given cell at level $k$ is $4^{-k}$.
The number of points contained in cells of \ringset thus follows a binomial distribution $B(n,p)$.
An upper bound for the probability $p$ is given by $\frac{8\cdot \sqrt{ns}}{n}$, thus a tail bound for a slightly different distribution $B(n,\frac{8\cdot \sqrt{ns}}{n})$
also holds for $B(n,p)$.
In the proof of \cite[Lemma~7]{Looz2015HRG} a similar distribution is considered.
Setting the variable $c$ to $8\sqrt{s}$, we see that the probability of \ringset containing more than $16\cdot \sqrt{sn}$ points is smaller than $0.7^{\sqrt{n}}$.

The subtrees in \subtreeset contain at most $s$ points by definition,
thus an upper bound for the number of points in these subtrees is given by $s\cdot 8\cdot \sqrt{n/s}$ (points not in \ringset) + $16\cdot \sqrt{sn}$ (points in \ringset).
This results in at most $24\cdot \sqrt{sn}$ points contained in \subtreeset with probability at least $1-0.7^{\sqrt{n}}$.
\qed
\end{proof}

The following Lemmas~\ref{lemma:subtree-aggregation-candidates} and~\ref{lemma:subtree-aggregation-cells} bound the number of examined candidates and examined quadtree cells and are used in the proof of Theorem~\ref{lemma:subtree-aggregation-complexity}.

\begin{lemma}
Let $T$ be a quadtree with $n$ points and $(q,f)$ a query pair.
The number of candidates examined by a query using subtree aggregation is in $O(|\neighborset(q,f)| + \sqrt{n})$ whp.
\label{lemma:subtree-aggregation-candidates}
\end{lemma}
\begin{proof}
For the analysis we consider each probability band $i$ separately.
As defined above, band $i$ contains points with a neighbor probability of $2^{-(i+1)}$ to $2^{-i}$.
Among the cells anchored in band $i$, some are original leaf cells and others are virtual leaf cells created by subtree aggregation.
The virtual leaf cells contain less than one expected candidate and thus less than $2^{i+1}$ points. The capacity of the original leaf cells is constant.
All the points in cells anchored in band $i$ have a probability between $2^{-(i+1)}$ and $2^{-i}$ to be a candidate.
Among the points in virtual or original leaf cells, some are in the same band their cell is anchored in, others are in higher bands.

We divide the set of points within cells anchored in band $i$ into four subsets:
\begin{enumerate}
 \item points in band $i$ and in original leaf cells
 \item points in band $i$ and in virtual leaf cells
 \item points not in band $i$ and in original leaf cells
 \item points not in band $i$ and in virtual leaf cells
\end{enumerate}

The points in the first two sets are unproblematic.
Since the probability that a point in these sets is a neighbor is at least $2^{-(i+1)}$, the probability for a given candidate to be a neighbor is at least $\frac{1}{2}$.
Due to Lemma~\ref{lemma:half-prob-set}, the number of candidates in these sets is in $O(|\neighborset(q,f)| + \log n)$ whp, which is in $O(|\neighborset(q,f)| + \sqrt{n})$ whp.

Points in the third set are in cells cut by the boundary between band $i$ and band $i+1$.
Since the probabilities are determined by the distance, this boundary is a circle and we can use Lemma~\ref{lemma:ring-root} to bound the number of points to $24\sqrt{n\cdot \mathrm{capacity}}$ with probability at least $1-0.7^{\sqrt{n}}$ for $n$ sufficiently large.
The mentioned capacity is the capacity of the original leaf cells.

Likewise, points in the fourth set are in virtual leaf cells cut by the boundary between bands $i$ and $i+1$.
A virtual leaf cell, which is an aggregated subtree, contains at most $2^{i+1}$ points, otherwise it would not have been aggregated.
Again, using Lemma~\ref{lemma:ring-root}, we can bound the number of points in these sets to $24\sqrt{n\cdot 2^{i+1}}$ points with probability at least $1-0.7^{\sqrt{n}}$.

We denote the union of the third and fourth sets with $\overhangset(q,f,i)$.
From the individual bounds derived in the previous paragraphs, we obtain an upper bound for the number of points in $\overhangset(q,f,i)$ of $24(\sqrt{n\cdot \mathrm{capacity}} + \sqrt{n\cdot 2^{i+1}})$ with probability at least $(1-0.7^{\sqrt{n}})^2$.
Simplifying the bound, we get that $|\overhangset(q,f,i)| \leq 24 \sqrt{n}\cdot(2^{(i+1)/2} + \sqrt{\mathrm{capacity}})$ with probability at least $1-2\cdot0.7^{\sqrt{n}}$.

Each of the points in $\overhangset(q,f,i)$ is a candidate with a probability between $2^{-i}$ and $2^{-(i+1)}$.
The candidates are sampled independently (see Step 2 of Lemma~\ref{lemma:independent-correct-probabilities}). 
While different points may have different probabilities of being a candidate and the total number of candidates does not follow a binomial distribution,
we can bound the probabilities from above with $2^{-i}$.

We proceed towards a Chernoff bound for the total number of candidates across all overhangs.
Let $X_i$ denote the random variable representing the candidates within $|\overhangset(q,f,i)|$ and let $X = \sum_{i=0}^{\infty} X_i$ denote the total number of candidates in overhangs.

The expected value $\mathbb{E}(X)$ follows from the linearity of expectations:
\begin{align}
\mathbb{E}(X) &= \sum_{i=0}^{\infty} \mathbb{E}(X_i)\\
&\leq \sum_{i=0}^{\infty} 24 \sqrt{n}\cdot(2^{(i+1)/2} + \sqrt{\mathrm{capacity}})\cdot2^{-i})\\
&= 24 \sqrt{n} \sum_{i=0}^{\infty} \sqrt{2}\cdot 2^{-i/2} + 2^{-i}\sqrt{\mathrm{capacity}}))\\
&= 24 \sqrt{n} ((2\sqrt{2}+2) + 2\sqrt{\mathrm{capacity}})
\end{align}

(Cells anchored in the band $\infty$, which has an upper bound $\overline{b}$ of zero for the neighborhood probability, do not have any candidates and can be omitted here.)

Since the candidates are sampled independently with a probability of at most $2^{-i}$, we can treat $X$ as a sum of independent Bernoulli random variables without loosing generality.
This allows us to use a multiplicative Chernoff bound~\cite{mitzenmacher2005probability} and we can now give an upper bound for the probability that the overhangs contain more than twice as many candidates as expected:
\begin{align}
\Pr(X > 2\mathbb{E}(X)) &\leq \left( \frac{e}{2^2} \right)^{\mathbb{E}(X)}\\
&= \left( \frac{e}{2^2} \right)^{24 \sqrt{n} ((2\sqrt{2}+2) + 2\sqrt{\mathrm{capacity}})}\\
&\leq \left( \frac{e}{2^2} \right)^{\sqrt{n}}\\
&\leq 0.7^{\sqrt{n}}
\end{align}

While the random variable $X = \sum_{i=0}^{\infty} X_i$ is written as an infinite sum, all but at most $n$ bands are empty, thus we are only applying the Chernoff bound over finitely many variables.
For each of the at most $n$ non-empty bands, we defined two tail bounds for the number of points in the overhang. 
Including this last bound, we thus have a chain of $2n+1$ tail bounds, each with a probability of at least $(1-0.7^{\sqrt{n}})$.
The event that any of these tail bounds is violated is a union over each event that a specific tail bound is violated.
With a union bound~\cite[Lemma 1.2]{mitzenmacher2005probability}, the probability that any of the individual tail bounds is violated is at most $(2n+1)0.7^{\sqrt{n}}$.
Since $\frac{1}{(2n+1)0.7^{\sqrt{n}}}$ grows faster than $n$ for $n$ sufficiently large,
we conclude that the total number of candidates is thus bounded by $O(|\neighborset(q,f)|) + 48\sqrt{n}((2\sqrt{2}+2) + 2\sqrt{\mathrm{capacity}})$ with probability at least $(1-1/n)$ for $n$ sufficiently large.
The leaf capacity is constant, thus the number of candidates evaluated during execution of a query $(q,f)$ is in $O(|\neighborset(q,f)| + \sqrt{n})$ whp.
\qed
\end{proof}

We proceed with an auxiliary result necessary for bounding the number of examined quadtree cells in a query:

\begin{lemma}
Let $T$ be a quadtree with $n$ points and $(q,f)$ a query pair.
The number of quadtree cells examined by a query using subtree aggregation is in $O((|\neighborset(q,f)| + \sqrt{n}) \log n)$.
\label{lemma:subtree-aggregation-cells}
\end{lemma}

%
To prove Lemma~\ref{lemma:subtree-aggregation-cells}, we first introduce another auxiliary lemma:
\begin{lemma}
Let $\mathbb{D}_R$ be a hyperbolic or Euclidean disk of radius $R$ and let $T$ be a polar quadtree on $\mathbb{D}_R$ 
containing $n$ points distributed according to Section~\ref{sub:notation}.
Let \uaqcset(q,f) be the set of unaggregated quadtree cells that have only (virtual) leaf cells as children (category~C2 in the proof of Lemma~\ref{lemma:subtree-aggregation-cells}).
 With a query using subtree aggregation, $|\uaqcset(q,f)|$ is in $O(|\neighborset(q,f)| + \sqrt{n})$ whp.
 \label{lemma:bound-unaggregated-quadtree-cells-with-leaf-children}
\end{lemma}
\begin{proof}
 Let $c\in \uaqcset(q,f,i)$ be such an unaggregated quadtree cell anchored in band $i$ that has only original or virtual leaf cells as children.
 It contains at least $2^i$ points and has four children, of which at least one is also anchored in band $i$.
 We denote this (virtual) leaf anchored in band $i$ with $l$.
 Since each child of $c$ contains the same probability mass (Lemma~\ref{lemma:node-cell-probabilities}), each point of $c$ is in $l$ with probability $1/4$:
\begin{equation}
 \Pr(p \in l | p \in c) = \frac{1}{4}.
\end{equation}

A point in $l$ is a candidate (in $l$) with probability $f(\dist(q,l))$, which is between $2^{-(i+1)}$ and $2^{-i}$ since $l$ is anchored in band $i$.
The probability that a given point $p\in c$ is a candidate in $l$ is then 
\begin{equation}
 \Pr(p \in l \wedge p \in \candidateset(q,f,i)\text{ | }p \in c) = \frac{1}{4}\cdot f(\dist(q,l)) \geq 2^{-(i+3)}
\end{equation}

Since the point positions and memberships in $\candidateset(q,f,i)$ are independent, we can bound the number of candidates in $l$ with a binomial distribution $B(|c|, 2^{-(i+3)})$.
The probability that $l$ contains no candidates is:
\begin{align}
 f(0, |c|, \frac{1}{8} \cdot 2^{-i}) &= (1-\frac{1}{8} \cdot 2^{-i})^{|c|}\\
 &\leq (1-\frac{1}{8} \cdot \frac{1}{2^i})^{2^i} 
\end{align}

Considered as a function of $i$, this probability is monotonically ascending.
In the limit of $2^i \rightarrow \infty$, it trends to $\exp(-1/8) \approx 0.88$, a value it never exceeds.
The probability that the cell $c$ contains at least one candidate is then above $1-\frac{1}{\sqrt[8]{e}} > 0.1$.

For each cell in \uaqcset, the probability that it contains at least one candidate is $> 0.1$.
Let $X$ be the random variable denoting the number of cells in \uaqcset that contain at least one candidate.
We define an auxiliary binomial distribution $B(|\uaqcset|, 0.1)$ and use a tail bound to estimate the number of cells in \uaqcset containing candidates.
Let $Y\propto B(|\uaqcset|, 0.1)$ be a random variable distributed according to this auxiliary distribution.

We use a tail bound from \cite{ArratiaGordon1989} to limit the probability that $Y < 0.05|\uaqcset|$ to at most $\exp(-|\uaqcset|/80)$.
Since $0.1$ was a lower bound for the probability that a cell contains a candidate, this tail bound also holds for $X$.
The probability that the set of \uaqcset contains at least $0.05|\uaqcset|$ many candidates is then at least $(1-\exp(-|\uaqcset|/80))$.

We continue with a case distinction:
\paragraph{If $|\uaqcset| \in \omega(\sqrt{n})$:}
The probability $(1-\exp(-|\uaqcset|/80))$ is then smaller than $(1-\exp(-\sqrt{n}/80))$, which is  $< 1/n$ for sufficiently large $n$.
Thus the number of examined quadtree cells during a query is then linear in the number of candidates.
Due to Lemma~\ref{lemma:subtree-aggregation-candidates}, this is in $O(|\neighborset(q,f)| + \sqrt{n})$.

\paragraph{If $|\uaqcset| \in o(\sqrt{n})$:} 
The cardinality $|\uaqcset|$ is trivially in $O(\sqrt{n})$.
\qed
\end{proof}

The proof of Lemma~\ref{lemma:subtree-aggregation-cells} then follows easily:
\begin{proof}
We split the set of examined quadtree cells into three categories:
\begin{itemize}
 \item leaf cells and root nodes of aggregated subtrees (C1)
 \item parents of cells in the first category (C2)
 \item all other (C3)
\end{itemize}
The third category (C3) then exclusively consists of inner nodes in the quadtree. When following a chain of nodes in category C3 from the root downwards, it ends with a node in category C2.
The size $|C3|$ is thus at most $O(|C2| \log n)$ whp, since the number of elements in a chain cannot exceed the height of the quadtree, which is $O(\log n)$ by Proposition~\ref{thm:quadtree-height}.

With a branching factor of 4, $|C1| = 4|C2|$ holds.

The number of cells in category C2 can be bounded using Lemma~\ref{lemma:bound-unaggregated-quadtree-cells-with-leaf-children} to $O(|\neighborset| + \sqrt{n})$ with high probability.
The total number of examined cells is thus in $O((|\neighborset| + \sqrt{n}) \log n)$.
\qed
\end{proof}

\pagebreak

\section{Visualizations of Experimental Results}
\begin{figure}[ht!]
\centering
\begin{tabular}{lr}
\includegraphics[width=0.55\linewidth]{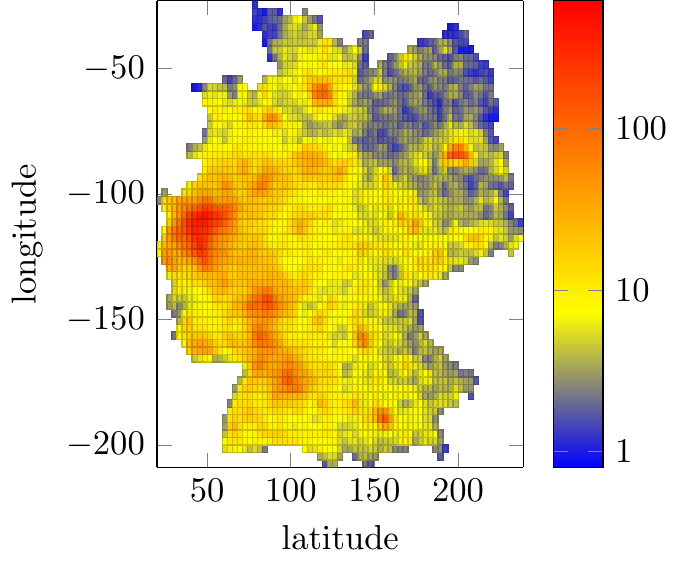}
\hspace{0.1cm}
&
\hspace{0.1cm}
\includegraphics[width=0.45\linewidth]{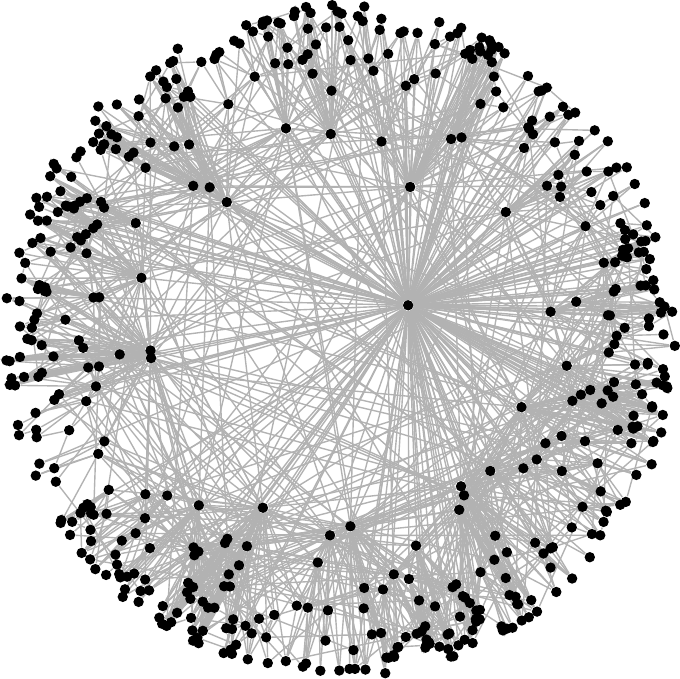}
\end{tabular}
 \caption{Left: Twenty-third time step of a simulated disease progression through Germany. 
The colors indicate the number of infected persons within a cell. 
Right: Random hyperbolic graph with 500 nodes and average degree 12.
}
 \label{plot:heatMap-and-hyperbolic-disk}
\end{figure}

\pagebreak

\section{Performance of Baseline Algorithm~\ref{algo:quadnode-probabilistic}}
\label{sec:baseline-performance}
\renewcommand{\aconst}{190}
\renewcommand{\bconst}{275}
\renewcommand{\cconst}{13}
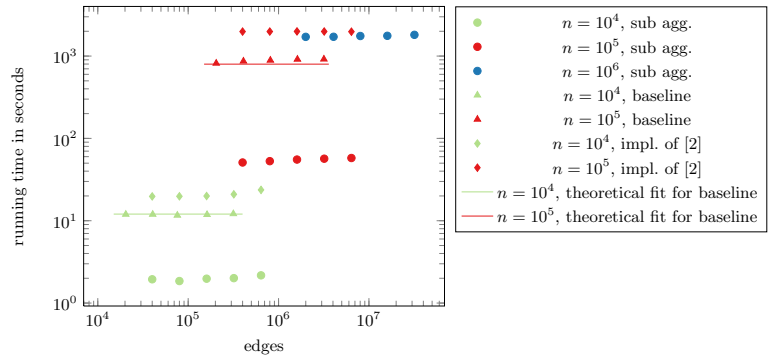
\begin{figure}[h!]
\centering
\begin{tikzpicture}[scale=0.7]
 \begin{axis}[xmode=log,ymode=log,xlabel=edges,ylabel=running time in seconds,legend entries={}, legend pos=outer north east]
  \addplot[scatter,only marks,
	   point meta = explicit symbolic,
	   scatter/classes={
	    f={draw=plottinggreen,fill=plottinggreen },%
	    g={draw=thirdhue,fill=thirdhue },%
	    h={draw=markedcolor,fill=markedcolor},
	    i={mark=triangle*, draw=plottinggreen, fill=plottinggreen},
	    j={mark=triangle*, draw, thirdhue, fill=thirdhue},
	    a={mark=diamond*, draw=plottinggreen, fill=plottinggreen},
	    b={mark=diamond*, draw, thirdhue, fill=thirdhue}}
	    ]
	    table[x=m, y=generationTime, meta=label]
          {plots/comparison-time-baseline.dat};
  \addlegendentry{$n=10^4$, sub agg.};
  \addlegendentry{$n=10^5$, sub agg.};
  \addlegendentry{$n=10^6$, sub agg.};
  \addlegendentry{$n=10^4$, baseline};
  \addlegendentry{$n=10^5$, baseline};
  \addlegendentry{$n=10^4$, impl. of \cite{Aldecoa2015}};
  \addlegendentry{$n=10^5$, impl. of \cite{Aldecoa2015}};
  \addplot[plottinggreen] expression[domain=15000:400000] {7.94*10^(-8)*10^8 + 4.1*10^(-4)*10^4};\addlegendentry{$n=10^4$, theoretical fit for baseline};
  \addplot[thirdhue] expression[domain=150000:3600000] {7.94*10^(-8)*10^10 + 4.1*10^(-4)*10^4};\addlegendentry{$n=10^5$, theoretical fit for baseline};
 \end{axis}
\end{tikzpicture}
 \caption{Comparison of running times to generate networks with $10^4$ to $10^6$ vertices. Generating a graph requires $n$ queries.
 Shown are running times of the baseline algorithm, queries using subtree aggregation and the implementation of \cite{Aldecoa2015}.
 The theoretical fit is given by the equation $T(n,m) = \left(7.94\cdot 10^{-8} n^2 + 4.1\cdot 10^{-4} n\right)$ seconds.
 The baseline algorithm is still faster than the previous implementation~\cite{Aldecoa2015}, but much slower than the improved query using subtree aggregation.
 }
 \label{plot:time-baseline-scatter}
\end{figure}

%
%
%

\pagebreak

\section{Fast RHG Generator vs Reference Implementation~\cite{Aldecoa2015}}
\label{sec:comparison-previous-impl}
\begin{figure}[!h]
\centering
\begin{tabular}{cc}
\includegraphics[width=0.41\linewidth]{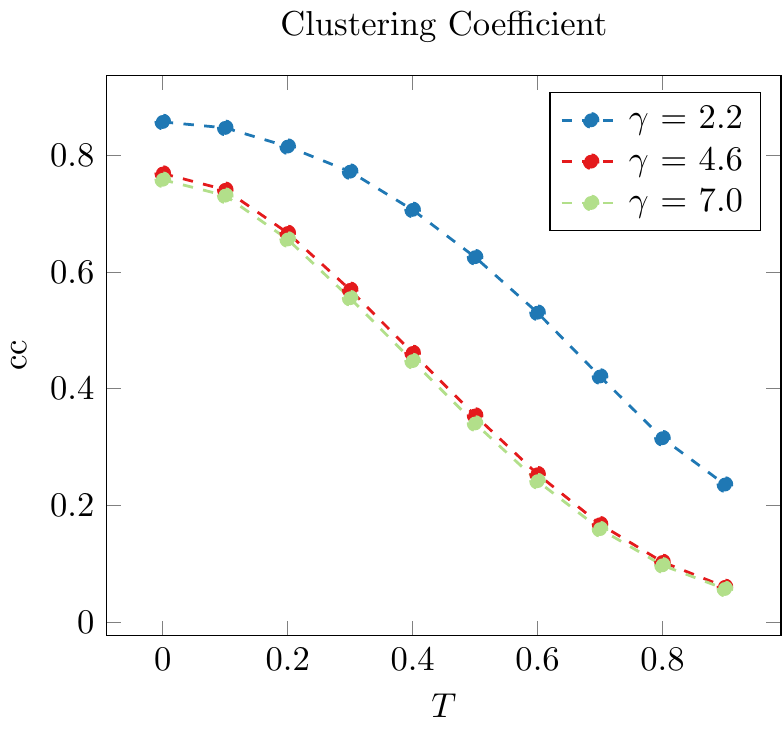} & \includegraphics[width=0.41\linewidth]{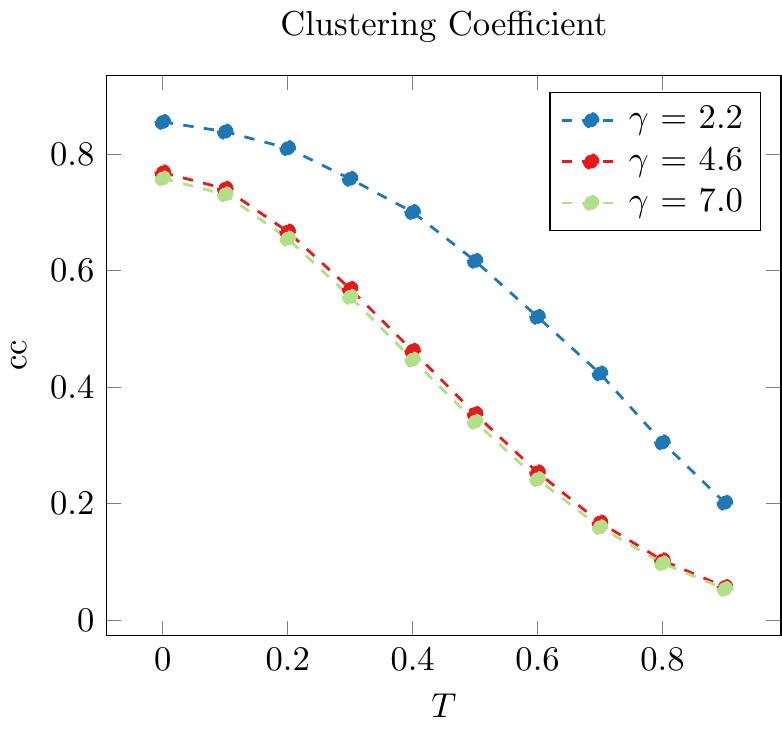}\\
\includegraphics[width=0.41\linewidth]{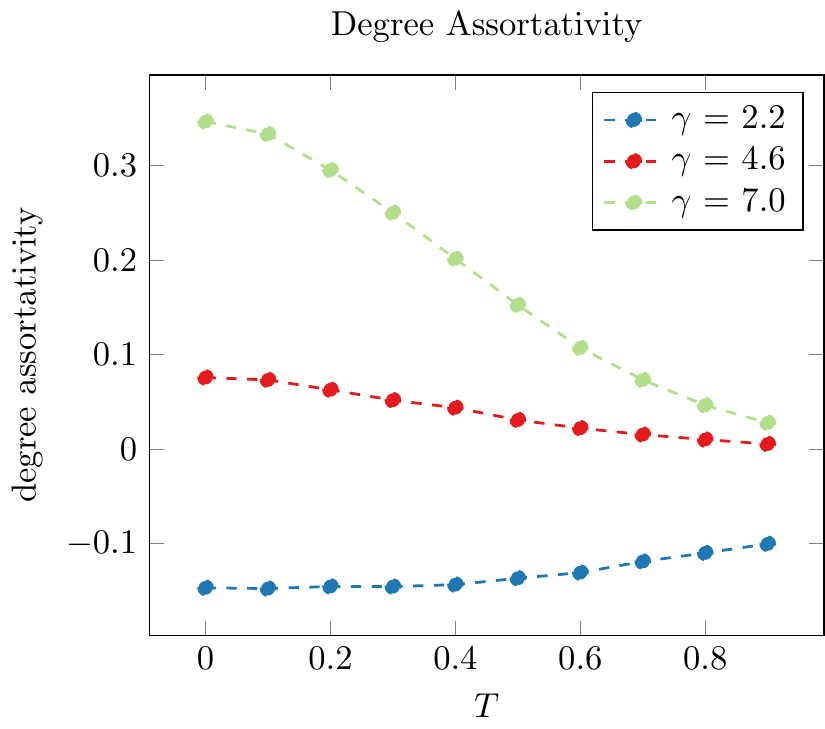} & \includegraphics[width=0.41\linewidth]{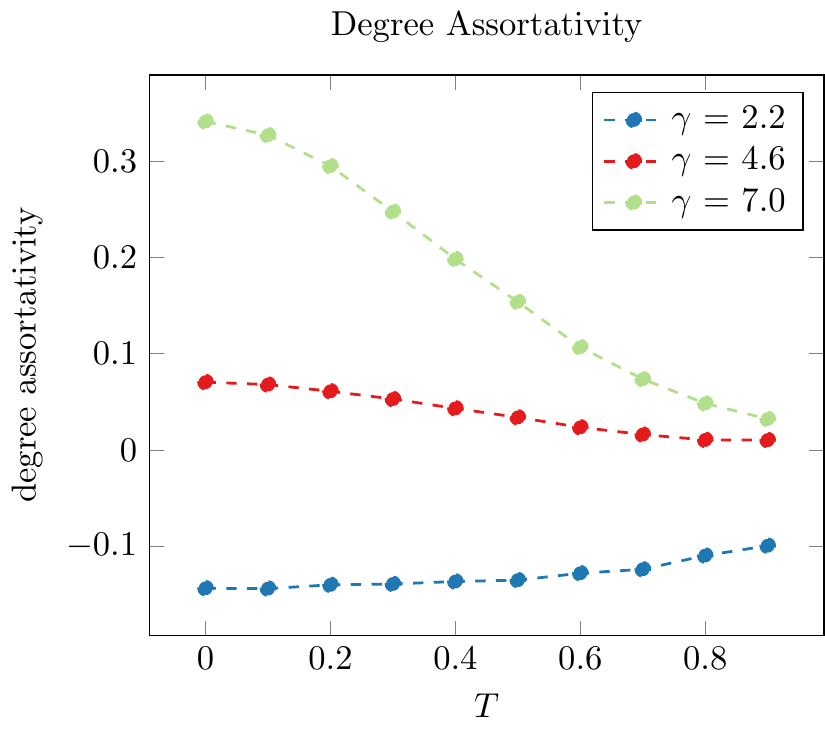}\\
\includegraphics[width=0.41\linewidth]{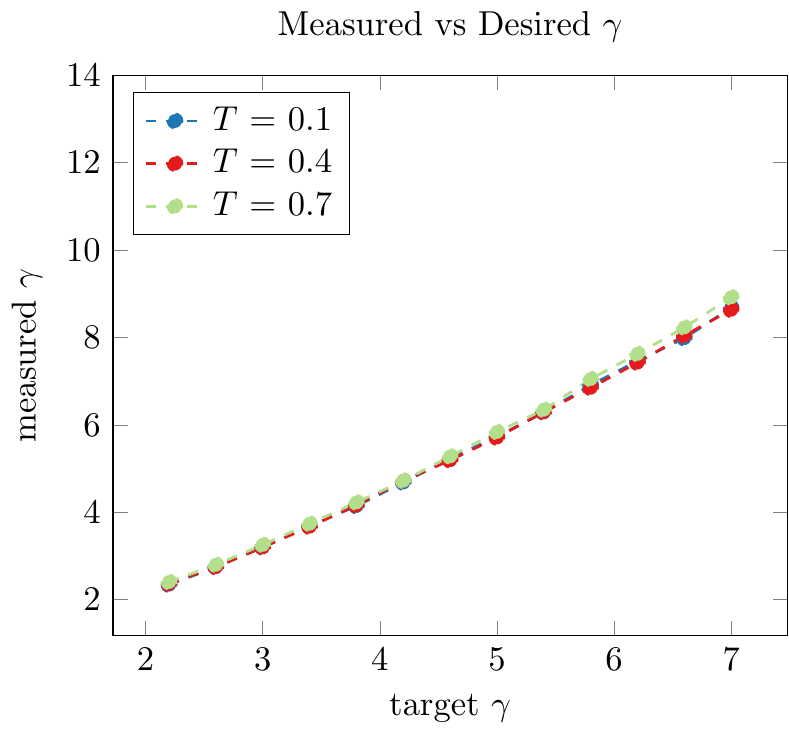} & \includegraphics[width=0.41\linewidth]{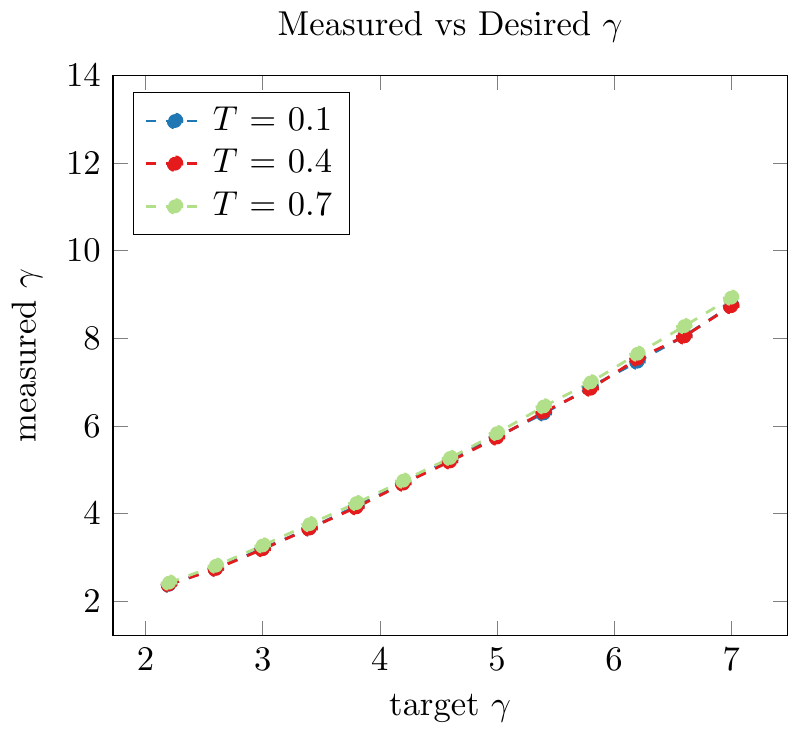}\\
\end{tabular} 
\caption{Comparison of clustering coefficients, degree assortativity and measured vs desired power-law exponent $\gamma$.
Shown are the implementation of~\cite{Aldecoa2015} (left) and our implementation (right).
The clustering coefficient describes the ratio of closed triangles to triads in a graph.
Degree assortativity describes whether vertices have neighbors of similar degree.
The degree distribution of random hyperbolic graphs follows a power law, whose exponent $\gamma$ can be adjusted.
In the degree distribution plot, the blue curve is almost always identical to the red curve and thus covered by it.
Values are averaged over 80 runs.}
\label{plot:properties-comparison-I}
\end{figure}
\end{document}